\documentclass[a4paper]{amsart}
\date{November 11, 2024}
\usepackage{amsthm,amssymb,bbm}
\usepackage{dsfont}
\usepackage{mathtools}
\newtheorem{theorem}{Theorem}
\newtheorem{lemma}{Lemma}
\newtheorem{corollary}{Corollary}
\newtheorem{definition}{Definition}

\def\tr{{\rm tr \,}}

\def\R{\mathbb{R}}
\def\C{{\mathbb C}}

\def\eps{\varepsilon}

\def\1{{\mathds{1}}}

\def\bp{{\mathbf p}}

\def\cA{\mathcal{A}}

\def\cD{\mathcal{D}}
\def\cE{{\mathcal E}}
\def\cS{{\mathcal S}}
\def\cU{{\mathcal U}}

\def\gA{\mathfrak{A}}
\def\gD{\mathfrak{D}}
\def\gH{\mathfrak{H}}
\def\gQ{\mathfrak{Q}}
\def\gS{\mathfrak{S}}
\def\gX{\mathfrak{X}}

\def\rd{\mathrm{d}}
\def\ri{\mathrm{i}}

\def\cz{\mathbb{C}} % complex numbers
\def\nz{{\mathbb N}}
\def\rz{\mathbb{R}} % real numbers

\def\atp{\bigwedge\limits} % antisymmetrisches Tensorprodukt
\def\sgn{\mathrm{sgn}} % Signum

\title[Ground State Energy of Heavy Atoms]{The Ground State Energy of
  Heavy Atoms: Leading and Subleading Asymptotics} \author{Long Meng
  and Heinz Siedentop}
\address{Mathematisches Institut\\
  Ludwig-Maximilians-Universit\"at M\"unchen\\
  Theresienstr. 39\\
  80333 M\"unchen\\
  Germany}
\email{long.meng@lmu.de \& h.s@lmu.de}

\begin{document}

\begin{abstract}
  We study atomic ground state energies for neutral atoms as the
  nuclear charge $Z$ is large in the no-pair formalism. We show that
  for a large class of projections defining the underlying Dirac sea --
  covering not only the physical reasonable cases but also ``weird''
  ones -- the corresponding no-pair ground state energy does not exceed
  the one of the Furry energy up to subleading order. An essential
  tool is the use and extension of S\'er\'e's results on 
  atomic Dirac-Hartree-Fock theory.
\end{abstract}

\maketitle
\section{Introduction and outline of paper}

In non-relativistic quantum mechanics there are canonical Hamiltonians
to describe atoms, molecules, and other many particle systems
involving electrons. These emerge as canonical quantization of a
classical Hamiltonian. These are successful when low momenta dominate
the physical effects. However, they are clearly inappropriate when
atoms of high atomic number $Z$ are present, since their inner nuclei
move with a substantial fraction of the velocity of light $c$. In fact
there are striking effects like the golden color of gold or the
functioning of lead-acid batteries which make a relativistic description
mandatory (see more details in Pyykkö's review \cite{Pyykko2012}).

From a fundamental physical point of view such systems should be
described by quantum electrodynamics. However, from a mathematical
point of view a consistent mathematical model is absent. Moreover,
although a perturbative evaluation is possible, in fact very
successful for few electrons, it is out of reach for many electron
systems. Moreover, an alternative ansatz that amounts to replacing the
non-relativistic kinetic energy by the free Dirac operator yields
Hamiltonians whose spectra are the entire real line. This phenomena,
known in the literature as ``continuum dissolution'' or
``Brown-Ravenhall disease'' (Sucher \cite{Sucher1984}, Datta and
Jagannathan \cite{DattaJagannathan1984}) is obviously unphysical. The
dilemma has been recognized shortly after the one-particle Dirac
equation was written down. Early attempts to resolve it go back -- at
least -- to Chandrasekhar \cite{Chandrasekhar1931} who used the
semiclassical quantization of the the classical energy momentum
relation $\sqrt{c^2\bp^2+c^4}$ (in the ultra-relativistic limit) to
determine the stability and instability of white dwarfs. It also has
been popular in the mathematical literature study relativistic effects
of Coulomb systems. However, the quantitative results become
inaccurate for larger atoms and fail completely for Radium ($Z=88$)
and elements with higher nuclear charge.

Starting with Brown and Ravenhall \cite{BrownRavenhall1951} (see also
Bethe and Salpeter \cite{BetheSalpeter1957}) Hamiltonians implementing
Dirac's idea of a sea that excludes negative energy states were
derived. These so called no-pair Hamiltonians were systematically
investigated by Sucher \cite{Sucher1980} taking into account the
polarization of the vacuum and its polarization by an external field
called the ``Furry picture'' (see Furry \cite{Furry1951}). Eventually
Mittleman \cite{Mittleman1981} derived a slightly different
Hamiltonian picture from quantum electro-dynamics by physical
arguments. Although, by definition the Mittleman energies
could be higher than the ground state energies they practically agree
in numerical calculations. Thus, it can conjectured, that the energies
of large neutral atoms agree up to small errors. We will -- in fact --
confirm this conjecture.

In Section \ref{mainresults} we will define the no-pair Hamiltonians
and Mittleman's variational principle in detail and formulate our main
result, i.e., the Mittleman energy and the Furry energy agree
asymptotically up to second order in the atomic number. As a byproduct
we obtain a justification of a S\'er\'e's variant of relativistic
Hartree-Fock theory of neutral atoms from a microscopic physical point
of view: Although a mathematically sound theory it was originally only
an ad hoc transcriptions from the non-relativistic setting. Here it
emerges as a limiting theory of a more fundamental theory.

Section \ref{sec:Sere} introduces a version of Dirac-Hartree
theory that is patterned in the spirit of S\'er\'e. In addition we
formulate an existence result for its minimizers. The functional
considered will be -- in some sense -- simpler than S\'er\'e's
original one, namely the exchange is dropped and replaced by a
Fermi-Amaldi correction (Fermi and Amaldi \cite{FermiAmaldi1934} and
Gombas \cite[\S 7]{Gombas1949}). Furthermore a generalized
$SU(2)$-invariance of the test density matrices is imposed which
results in spherical symmetry of the associated densities.

This functional will be the essential tool in the proof of the main
theorem. Moreover it allows us to describe admissible projections in
Mittleman's variational characterization of ground state energy. The
definition of admissibility and sufficient criteria for admissibility
will be given in Subsection \ref{subsec:admissibility}. Subsection
\ref{examples} will give examples of admissible potentials.

Theorem \ref{th:1} will be proven in Section \ref{sec:mittleman}. As
mentioned above the proof uses the S\'er\'e type functional and some
new properties of its minimizer $\gamma_*$.  Although we will follow
the general strategy of S\'er\'e \cite{Sere2023}, we need to extend
his existence result in a way suitable for our application and derive
some new properties of the resulting minimizers.

Appendix \ref{sec:existence} indicates the existence proof of
minimizers of the functional with particular emphasis on the new
aspects not contained in previous work. Appendix \ref{sec:Z7/3} gives
some properties of these minimizers that are essential for our proof.

\section{Notations and main results}\label{mainresults}
We begin by introducing some notation that will allow us to formulate
our result.  We write
\begin{equation}
  \label{DZ}
  D_{c,Z}=c\bp\cdot\boldsymbol\alpha+c^2\beta-{Z\over|x|}
\end{equation}
for the Dirac operator of an electron in the field of nucleus of
atomic number $Z$. Here $\bp:=-\ri\nabla$ denotes the momentum
operator and $\boldsymbol\alpha$ and $\beta$ are the four Dirac
matrices in standard form.  Eventually, we write $c$ for the velocity
of light.  In Hartree units -- used throughout this paper -- its
physical value is about $137$. However, we keep $c$ and also $Z$ as
parameters but fix their quotient
\begin{equation}
  \label{kc}
  \kappa:=Z/c\in(0,1)
\end{equation}
throughout the paper.

The operators $D_{c,Z}$ are selfadjoint operators in
$\gH:=L^2(\rz^3:\cz^4)$ defined in the sense of Nenciu
\cite{Nenciu1976}. Their form domains are $\gQ:=H^\frac12(\rz^3:\cz^4)$.
Morozov and M\"uller \cite{MorozovMuller2015} characterized also their
domain showing which spinors have to be added to $H^1(\rz^3:\cz^4)$.
They exhibit the scaling
\begin{equation}
  \label{eq:scale}
  D_\kappa:=D_{1,c/Z}\cong D_{c,Z} /c^2
\end{equation}
In this notation we have for $\kappa\in[0,1)$
\begin{equation}
  \label{eq:D-kappa-form}
  C_\kappa|D_0|\leq|D_\kappa|\leq(1+2\kappa) |D_0|
\end{equation}
with
\begin{align}
  \label{5}
  C_\kappa:=& \max\left\{{d_\kappa\Upsilon_\kappa\over1+d_\kappa},1-2\kappa\right\},\
  d_\kappa:= (1-\pi\Upsilon_\kappa\cot(\pi\Upsilon_\kappa/2)/2)\eta_\kappa,\\
  \Upsilon_\kappa:=&\sqrt{1-\kappa^2},\
  \eta_\nu:=
  \begin{cases}
    {\sqrt{9+4\kappa^2}-4\kappa\over3(1-2\Upsilon_\kappa\cot(\pi\Upsilon_\kappa/2))} & \kappa\in[0,1)\setminus\{\sqrt3/2\}\\
    {1\over\sqrt3(\pi-2)} & \kappa=\sqrt3/2
  \end{cases}.
\end{align}
(Morozov and M\"uller \cite[Theorem I.1]{MorozovMuller2017}, see also
Frank et al \cite[Corollary 1.8]{Franketal2021}). 

The energy form of an atom for states
$\Psi\in\mathcal{S}(\R^{3N}:\C^{4^N})$ is
\begin{equation}
  \cE_{c,Z,N}[\Psi]:=\left\langle \Psi,\left(\sum_{n=1}^N (D_{c,Z}-c^2)_n+\sum_{1\leq m<n\leq N}\frac{1}{|x_n-x_m|}\right)\Psi\right\rangle
\end{equation}
From now on, we will assume that $N,Z\geq1$ and, unless stated
otherwise, that we have neutrality, i.e., $N:=Z$ and accordingly we
will suppress the indices $N$ and also $c$ with all energy functionals
and energies, since $\kappa$ is fixed and $c=\kappa Z$ under our
general assumption \eqref{kc}.  Since electrons are fermions, we will
restrict the domain to antisymmetric spinors $\Psi$. Then the
corresponding $N$-electron spaces are
\begin{equation}
  \label{eq:n}
  \gH^N:=\atp_{n=1}^N \gH,\qquad \gQ^N:=\atp_{n=1}^N\gQ.
\end{equation} 
Moreover, due to the unboundedness of the Dirac operator from below,
Dirac postulated that negative energies states, the Dirac sea, are not
accessible to electrons. Brown and Ravenhall \cite{BrownRavenhall1951}
and later Sucher \cite{Sucher1984,Sucher1987}, extending this idea,
formulated this mathematically by requiring that the one-electron
space is given by the positive spectral subspace of a suitable Dirac
operator.

In the following $A$ will be a -- possibly $Z$ dependent -- perturbation
of the hydrogenic operator. Although we will specify the exact class
later, we will always require three basic properties of $A$:
\begin{description}
\item[Selfadjointness] The $A$ is symmetric operator on
  $\cS(\rz^3:\cz^4)$ and $D_{c,Z}+A$ has selfadjoint extension
  $D_{c,Z}^A$ with form domain $\gQ:=H^\frac12(\rz^3:\cz^4)$ in the
  sense of Nenciu.
\item[Invariance of $H^\frac12$]
\begin{equation}
  \label{eq:P}
  \gQ_A:=P_{Z,A}H^\frac12(\rz^3:\cz^4)\subset H^\frac12(\rz^3:\cz^4)
\end{equation}
where $P_{Z,A}:= \1_{(0,\infty)}(D_{c,Z}+A)$.
\item[Boundedness from above] 
  On $\gQ$
  \begin{equation}
    \label{boundedness}
    P_A^\perp D_{c,Z}P_A^\perp \leq c^2.
  \end{equation}
  Physically speaking this ensures that positronic states in the
  picture defined by the perturbed hydrogenic Dirac operator
  $D_{c,Z}+A$ still have an expectation of unperturbed Hamiltonian
  that does not exceed the spectral gap of the essential spectrum. In
  fact one might expect zero on the right side whereas the above
  condition allows for more general potentials $A$.
\end{description}
The corresponding $N$-electron Hilbert spaces and form domains are
\begin{equation}
  \gH^N_A:=\atp_{n=1}^N  P_{Z,A}(\gH),\qquad \gQ^N_A:=\atp_{n=1}^N P_{Z,A}(\gQ).
\end{equation}

Typical choices for such operators:
\begin{itemize}
\item {\bf Free picture (Brown--Ravenhall \cite{BrownRavenhall1951})}
  Here $A(Z)= Z/|x|$.
\item {\bf Furry picture (Furry \cite{Furry1951})} Here
  $A=0$.
\item {\bf Intermediate or Fuzzy picture (Mittleman
    \cite{Mittleman1981})} In this case, $A$ is picked as a mean-field
  potential of the electrons. An optimal choice, which
  depends on the two-particle density matrix, was suggested by
  Mittleman. This leads to a nonlinear equation, which has been
  studied numerically with great success in quantum chemistry, see
  \cite{Saue2011}.
\end{itemize}
For a given $A$, the ground state energy of the relativistic
Coulomb system is 
\begin{equation}
  E_{Z,A}:=\inf\{\cE_{Z}[\Psi]\big|\Psi\in \gQ^N_A, \|\Psi\|\leq 1\}.
\end{equation}

In Sucher's terminology a choice of a potential $A$ defines a
picture. However, different pictures might lead to different ground
state energies, in particular they might not by unitarily
equivalent. For example, the ground state of the Furry picture is
larger than the Brown--Ravenhall one:
\begin{itemize}
\item The ground state energy in Furry picture was found in
  \cite{HandrekSiedentop2015}. For large $Z$  it behaves as
  \begin{equation}
    \label{eq:Scott-Furry}
   E^{\rm F}_{Z}:=E_{Z,0}= C^\mathrm{TF}Z^{\frac73}+C_\mathrm{Scott}^{\rm F}Z^2+o(Z^2)
 \end{equation}
 with $C^\mathrm{TF}<0$ being the Thomas--Fermi ground-state energy of
hydrogen, and
\begin{equation}
  C_\mathrm{Scott}^{\rm F}:= \frac{1}{\kappa^2}
  \sum_{n=1}^\infty(\lambda_n^{\rm D}-\lambda_n^{\rm S})+\frac12
\end{equation}
where $\lambda_n^{\rm D}$ is the $n$-th eigenvalue of $D_\kappa$.
\item Frank et al. \cite{Franketal2009} showed that the ground state
  of the Brown--Ravenhall operator up to its critical coupling
  constant $\kappa<(0,2/(\pi/2+2/\pi))$ behaves for large $Z$ as
\begin{equation}
   E^{\rm BR}_{Z}:=E_{Z,Z|\cdot|^{-1}} = C^\mathrm{TF}Z^{\frac73}+C_\mathrm{Scott}^{\rm BR}Z^2+o(Z^2)
\end{equation}
with
\begin{equation}
    C_\mathrm{Scott}^{\rm BR}:= \frac{1}{\kappa^2}\sum_{n=1}^\infty(\lambda_n^{\rm BR}-\lambda_n^{\rm S})+\frac{1}{2}
\end{equation}
where $\lambda_n^{\rm BR}$ is the $j$-th eigenvalue of
$\1_{(0,\infty)}(D_0)D_\kappa\1_{(0,\infty)}(D_0)$.
\end{itemize}
The minimax principle for the Dirac eigenvalue problem (see, e.g.,
\cite{GriesemerSiedentop1999,Dolbeaultetal2000O} shows that
$\lambda_n^{\rm BR} < \lambda_n^{\rm D}$. Hence
$C_\mathrm{Scott}^{\rm BR}< C_\mathrm{Scott}^{\rm F}$ and thus
$E_{Z}^{\rm BR}<E_{Z}^{\rm F}$.

According to Mittleman \cite{Mittleman1981} the physical ground-state
energy is obtained by taking the supremum over all projections of the
form $\1_{(0,\infty)}(D_{c,Z}+A)$.  However, Mittleman does not
specify the class of potentials over which the supremum should be
taken to give physically meaningful results. In fact there are obvious
although non reasonable counterexamples -- like $A=c^2$. Nevertheless,
he computes formally the Euler-Lagrange equation both in the fully
interacting setting as well as in the Hartree-Fock setting suggesting
that the approximate maximizers should be of mean-field type.

The aim of this paper is twofold:
\begin{itemize}
\item Address the lack of a suitable class of potentials in Mittleman's
  argument, i.e., to specify a reasonable class of admissible
  potentials.
\item Show that Mittleman's energy agrees up to subleading order with the Furry energy.
\end{itemize}
Our main result is
\begin{theorem}
  \label{th:1}
  Assume that $A$ is an admissible family of operators (see Definition
  \ref{def:A}). Then
  \begin{equation}\label{eq:Mittleman}
    \boxed{\sup_{A\ \text admissible}\lim_{Z\to\infty}{E_{Z,A(Z)} -C^\mathrm{TF}Z^\frac73-C_\mathrm{Scott}^{\rm F}Z^2\over Z^2}=0.}
  \end{equation}
 \end{theorem}
\begin{proof}
  The upper bound is provided by the following theorem. The lower
  bound from \cite{HandrekSiedentop2015} and the admissibility of the
  Furry picture, i.e., $A=0$.
\end{proof}

\begin{theorem}
  \label{th:mittleman}
  Assume $A$ is an admissible family of operators (see Definition
  \ref{def:A}). Then
  \begin{equation}
    \boxed{E_{Z,A(Z)}\leq E_Z^\mathrm{F}+o(Z^2)=C^\mathrm{TF}Z^\frac73+C_\mathrm{Scott}^{\rm F}Z^2+o(Z^2).}
  \end{equation}
\end{theorem}

The admissibility condition does not only allow for the above typical
choices (Free, Furry, and intermediate picture) but also for
potentials that are vastly different from typical mean fields like point
Coulomb potentials and others. We will discuss some examples in Section
\ref{examples}.

\section{Relativistic Hartree Theory \`a la  S\'er\'e and admissible potentials}
\label{sec:Sere}
We begin with a study of the Dirac-Hartree functional in
the sense of S\'er\'e. This will allow us to characterize the
admissible potentials.

\subsection{Hartree  theory with Amaldi correction} 
We write
\begin{equation}
\gX:=\{\gamma\in\gS^1({\gH})| 0\leq\gamma\leq1, \|\gamma\|_{\gX_c}:=\tr[(c^4+c^2\bp^2)^\frac14\gamma(c^4+c^2\bp^2)^\frac14]<\infty\}.
\end{equation}

We write $\mu_1\geq\mu_2\geq... $ -- counting multiplicities
-- for the eigenvalues of $\gamma\in \gS^1(\gH)$ and $\psi_1,\psi_2,...$ for
the associated orthonormal eigenspinors yielding
\begin{equation}
  \gamma=\sum_{n\in \nz}\mu_n |\psi_n \rangle\,\langle \psi_n|.
\end{equation}
Moreover we write $\rho_\gamma$
\begin{equation}
  \label{dichte}
\rho_{\gamma}(x):=\sum_{n\in\nz}\mu_n |\psi_n(x)|^2.
\end{equation}
(Note that the notation $\rho_K$ can be generalized to any selfadjoint
operator $K$ by replacing the right hand side of \eqref{dichte} by the
eigenvalues and eigenfunctions of $K$ provided the right hand side
converges absolutely.)

For any $\gamma\in\gX$ the Dirac-Hartree operator for an atom is
\begin{equation}
    D_{c,Z,\rho}:= D_{c,Z}+\rho*|\cdot|^{-1}
\end{equation}
and the Dirac-Hartree functional with and without Fermi-Amaldi
correction reads 
\begin{align}
  \mathcal{E}^{\rm HA}_{Z,N}(\gamma):=&\tr[(D_{c,Z}-c^2)\gamma]+(1-N^{-1})\cD[\rho_{\gamma}]\\
  \mathcal{E}^{\rm H}_Z(\gamma):=&\tr[(D_{c,Z}-c^2)\gamma]+\cD[\rho_{\gamma}]
\end{align}
with the Hartree term 
\begin{align*}
    \cD[\rho_{\gamma}]:= \tfrac12\int_{\R^3}\int_{\R^3}\frac{\rho_{\gamma}(x)\rho_{\gamma}(y)}{|x-y|}\rd x \rd y.
\end{align*}
Originally the Fermi-Amaldi correction \cite{FermiAmaldi1934} was
introduced to eliminate the self-interaction of electrons in the
Hartree term under the constraint $\int\rho_\gamma=N$. For us,
this will have the technical consequence that threshold eigenvalues
can be disregarded in the proof of the existence of minimizers in
Appendix \ref{sec:existence}.

We will consider density matrices whose density is spherically
symmetric. In fact we will require a $SU(2)$ invariance. Then 
$a\in SU(2)$ determines unique rotation $R_a\in SO(3)$ by the relation
\begin{equation}\label{eq:A-RA-identity}
  R_ax \cdot \boldsymbol\sigma = a\circ(x\cdot\boldsymbol\sigma)\circ a^{-1}.
\end{equation} 
$x\in \rz^3$, since the map $a\mapsto R_a$ is surjective because
$SO(2)\cong SU(2)/\mathbb{Z}_2$ (see, e.g., Woit \cite{Woit2017}) any
rotation can be realized.

Now we will define a unitary representation $U_a$ of $SU(2)$ on $\gH$:
\begin{equation}
  \left(U_a
  \begin{pmatrix}
    u\\v
  \end{pmatrix}\right)(x)
:= \begin{pmatrix} (au) (R_a^{-1}x)\\
(av)((R_a^{-1}x)
\end{pmatrix}
\end{equation}
for $(u,v)^t\in L^2(\rz^3:\cz^2)\oplus L^2(\rz^3:\cz^2)\equiv \gH$. The invariant density matrices are 
\begin{equation}
  \Gamma_{N}:=\{\gamma \in\gX| \tr(\gamma)\leq N, \forall_{a\in SU(2)} U_a^{-1}\gamma U_a=\gamma \}.
\end{equation}
Moreover set
\begin{equation}
    \Gamma_{Z,N}:=\{\gamma \in \Gamma_{N}| P_{Z,\frac{N-1}N\rho_\gamma}\gamma P_{Z,\frac{N-1} N\rho_\gamma}=\gamma\}
\end{equation}
with
\begin{equation}
    P_{Z,\rho}= \mathbbm{1}_{(0,\infty)}(D_{c,Z,\rho}).
  \end{equation}
  Note also that for any $\gamma\in\Gamma_N$ and $a\in SU(2)$ we have
  $[P_{Z,\rho_\gamma},U_a]=0$.

  In the spirit of S\'er\'e \cite{Sere2023} we introduce
\begin{equation}
  \begin{split}
  \cE_{Z,N}^\mathrm{S}: \Gamma_{Z,N}&\to \rz\\
  \gamma&\mapsto\cE^\mathrm{HA}_{Z,N}(\gamma)
  \end{split}
\end{equation}
i.e., the restriction of the Dirac-Hartree functional with
Fermi-Amaldi correction (Fermi and Amaldi \cite{FermiAmaldi1934}) to
$\Gamma_{Z,N}$. We call it as a short hand ``S\'er\'e functional''
(instead of the correct but somewhat lengthy name ``
atomic Dirac-Hartree functional with a Fermi-Amaldi correction and a
constraint introduced by S\'er\'e restricted to $SU(2)$ invariant
density matrices''). Furthermore we write
\begin{equation}
  \label{eq:min}
  E^{\rm S}_{Z,N}:=\inf \cE^\mathrm{S}_{Z,N}(\Gamma_{Z,N})
\end{equation}
for its infimum and call it the S\'er\'e energy. If $N=Z$, we will
drop the double index and simply write
\begin{equation}
  \label{eq:simple}
  E^\mathrm{S}_Z,\ \cE^\mathrm{S}_Z,\ \cE^\mathrm{HA}_Z,....
\end{equation}
  
The existence of minimizers for the Dirac-Hartree-Fock problem with
the exchange term has been treated \cite[Theorem 1.2]{Sere2023} up to
$\kappa\leq 22/137$ for $N=Z$ (see \cite[Remark
1.4]{Sere2023}). We will improve this for the model at hand to the
full range \eqref{kc} stated in the introduction.
  
We will need a result of Fournais et al. \cite[Theorem
4]{Fournaisetal2020}:
\begin{lemma}
  \label{lem:DF-ope}
  Assume $\rho$ nonnegative, spherically symmetric, and
  $\sqrt\rho\in H^\frac12(\rz^3)$, and $N:=\int\rho\leq Z$. Then there
  exist a constant $C_\kappa'>0$ such that
  \begin{equation}
    \label{2.15}
    C_\kappa'|D_{c,0}|\leq |D_{c,Z,\rho}|\leq (1+4\kappa)|D_{c,0}|.
  \end{equation}
\end{lemma}
\begin{proof}
  The lower bound were merely a transcription of \cite[Theorem
  4]{Fournaisetal2020} if it were not for the extended range of
  $\kappa<1$. Since, in our case, $\rho$ is spherically symmetric, it
  allows for $\nu_0=1$ in \cite[Equation (14)]{Fournaisetal2020} by
  \cite[Lemma 2]{Fournaisetal2020}. In particular the constant
  $C_\kappa'$ is the constant defined in \cite[Equation
  (19)]{Fournaisetal2020} but now with $\nu_0=1$ and $\nu=\kappa$.

  The upper bound is merely an estimate of the mean-field by Hardy's
  inequality combined with the monotonicity of the square root (Kato
  \cite[Theorem 2]{Kato1952}).
\end{proof}

To formulate our existence theorem we need to recall some constants:
We write $C^\mathrm{ret}_\kappa$ for the constant $C_{\kappa,\nu}$
with $\nu=\kappa$ in the first displayed equation of the proof of
\cite[Lemma 17]{Fournaisetal2020}. Moreover we write
$C^\mathrm{D}\geq 1.63/4^{1/3}\geq 1.02$ for the constant in
Daubechies' inequality \cite[Inequality (3.4)]{Daubechies1983} and
$C^\mathrm{HLT}:= \sqrt[3]{256/(27\pi)}\leq 1.45$ for the constant in
the Hardy-Littlewood-Sobolev inequality. Finally, we define
  \begin{equation}\label{eq:C-ex}
    C_\kappa^\mathrm{ex}:= 8\kappa\max\left\{\kappa^2\left({ C^{\rm HLS}\over C_\kappa' C^{\rm D}}\right)^3, 8\left({ C_{\kappa}^{\rm ret}\over C_\kappa'}\right)^6\right\}.
  \end{equation}

\begin{theorem}\label{th:DF-existence}
  For $Z> C_\kappa^\mathrm{ex}$ and $N\leq Z$ the S\'er\'e functional
  $\cE^\mathrm{S}_{c,Z,N}$ has a minimizer which we denote by
  $\gamma_*$.
\end{theorem}
As a short hand we call $\gamma_*$ a S\'er\'e minimizer and introduce
the following notation
\begin{equation}
  \label{*}
  \rho_*:=\rho_{\gamma_*},\ \phi_*:=(1-N^{-1})\rho_**|\cdot|^{-1},\ D_*:=D_{c,Z,{N-1\over N}\rho_*}, P_*:=\1_{(0,\infty)}(D_*)
\end{equation}
for the density of the S\'er\'e minimizer $\gamma_*$, the associated
mean-field potential $\phi_*$, the associated mean-field operator Dirac
$D_*$, and the positive spectral projection of $D_*$. Analogously to
the above we address these as S\'er\'e density and S\'er\'e operator.
\begin{theorem}
  \label{th:property}
  Under the assumptions of Theorem \ref{th:DF-existence} we have
  \begin{align}
      \gD(D_*)=&\gD(D_{c,Z}),\ \gQ(D_*)=\gQ,\\
       \sigma_\mathrm{ess}(D_*)=&\sigma_\mathrm{ess}(D_{c,0}) = (-\infty,-c^2]\cup[c^2,\infty),\\
    \label{discrete}
    \sigma_d(D_*) \subset&(\sqrt{c^4-c^2Z^2},c^2)\ \text{and}\ |\sigma_d(D_*)|=\infty,\\
    \tr(\gamma_*)=&N.
  \end{align}
  Moreover, writing $\lambda_1\leq\lambda_2\leq...$ for the positive
  eigenvalues of $D_*$ repeated according to their multiplicity and 
  $\psi_1, \psi_2,...$ for a system associated orthonormal eigenvectors,
  there is a $K\in \nz$ with $\lambda_{K-1}<\lambda_K$ and a finite
  rank operator $\delta$ with $0\leq \delta<1$ and
  $D_*\delta=\lambda_K\delta$ such that
\begin{equation}\label{eq:gamma}
  \gamma_*= \sum_{k=1}^{K-1}|\psi_k\rangle\langle\psi_k|
  +\delta.
  \end{equation} 
\end{theorem}

Its proof is basically due to S\'er\'e's \cite{Sere2023}. However,
because of some modifications based on \cite{Fournaisetal2020} and
\cite{Meng2023}, we give some details in Appendix
\ref{sec:existence}.

The proof of \eqref{discrete} uses that $D_{c,Z}\leq D_*$ and
$\lambda_1(D_{c,Z})=\sqrt{c^4-Z^2c^2}$ (Gordon \cite{Gordon1928}, see
also Bethe \cite[Formula (9.29)]{Bethe1933}).

Note that we do not claim uniqueness of S\'er\'e minimizers which is
in contrast to non-relativistic reduced Hartree theory where the
minimizer is unique by convexity. Since $\Gamma_{Z,N}$ is not a convex
set we leave the uniqueness question open.

The S\'er\'e functional will on the one hand serve as our starting
point for finding the Mittleman energy up to errors of the order
$o(Z^2)$. On the other hand it allows us to connect to the Furry
picture. In fact Fournais et al. \cite{Fournaisetal2020} showed that
the S\'er\'e energy and the Furry energy agree up to this order at
least if $\kappa<2/(\pi/2+\pi/2)$. We are able to extend this to the
model at hand for the full range \eqref{kc},i.e., to $0<\kappa<1$:
\begin{theorem}
  \label{th:DF-asymptotic}
  For $Z$ large 
     \begin{equation}
         E^{\rm S}_{Z}= E^\mathrm{F}_Z+o(Z^2). 
     \end{equation}
\end{theorem}
For technical reasons we will not only need bounds on the total
S\'er\'e energy but also on its potential and kinetic parts
separately.
\begin{theorem}\label{th:DF-Z7/3}
  Assume $N=Z$. Then for $Z$ large enough, any S\'er\'e minimizer
  $\gamma_*$ in \eqref{eq:min} satisfies
    \begin{align}\label{eq:Z7/3-1}
      Z\tr[|\cdot|^{-1}\gamma_*]&=O(Z^\frac73),\ \|\phi_*\|=O(Z^\frac43),\ 
      \cD[\rho_*]=O(Z^\frac73),\\
   \label{eq:Z7/3-2}
      \tr[(|D_{c,0}|-c^2)\gamma_*]&=O(Z^\frac73),\ % \label{eq:Z7/3-3}
                                    \tr[|\bp|\gamma_*]\leq O(Z^\frac53)\\
      \label{eq:Z7/3-3}
        \cD[\rho_*-\rho^{\rm TF}_Z]&=o(Z^2)
    \end{align}
    where $\rho^{\rm TF}_Z$ is the minimizer of the atomic
    Thomas-Fermi functional with atomic number $Z$.
\end{theorem}
The proof will be given in Appendix \ref{sec:Z7/3}.

\subsection{Admissible potentials}
\label{subsec:admissibility}
We will now describe the pictures considered, i.e., to specify the
allowed pictures.
\begin{definition}\label{def:A}
  We call $A$ admissible and write $A\in \gA$, if and only if $A$ is a
  family of symmetric operators $A(Z)$ with $Z\in[Z_0,\infty)$ for
  some $Z_0>0$  such that for all $Z\in[Z_0,\infty)$ we have
  $\gQ(A(Z))\supset \gQ$, that $ D_{c,Z} + A(Z)$ has a
  distinguished self-adjoint extension $D_{c,Z}^{A(Z)}$ in the sense
  of Nenciu, and the the following inequalities hold 
   \begin{align}
    \label{eq:Birman}
     &-1\lesssim B_{A(Z)}:=|D_{c,0}|^{-\frac12}A(Z)|D_{c,0}|^{-\frac12}\lesssim1,\\
    \label{eq:ope-A}
    &|D_{c,0}|\lesssim |D_{c,Z}^{A(Z)}| \lesssim  |D_{c,0}|,\\
    \label{eq:trace*op'}
     &P_{Z,A(Z)}^\perp D_{c,Z}P_{Z,A(Z)}^\perp \leq c^2,\\
     & \label{eq:trace*op}
       \tr\left(A(Z)\gamma_*A(Z)|D_{c,0}|^{-1}\right)= o(Z^\frac{12}5),\ Z\to\infty.
  \end{align}
  If $A$ fulfills all the above conditions except \eqref{eq:trace*op}
  we write $A\in\tilde\gA$.
\end{definition}

From now we will -- in abuse notation -- also use $D_{c,Z}+B$ for its
selfadjoint extension in the sense of Nenciu instead of $D_{c,Z}^B$.

We note that \eqref{eq:trace*op'} is \eqref{boundedness} and
\eqref{eq:ope-A} implies $\gQ_A\subset\gQ$ which is \eqref{eq:P}.

The following Lemma gives a sufficient condition for an operator to be
in $\tilde\gA$. 
\begin{lemma}\label{lem:ope-A-s}
  Suppose $A(Z)$ is a family of
  symmetric operators in $\gH$ with $Z\geq Z_0$ and for any such $Z$
  there exist constants $\epsilon_Z, M_Z\in [0,\infty)$ such that for
  all $\psi\in \gQ$
  \begin{equation}
    \label{eq:form}
    \|A(Z)\psi\| \leq \epsilon_Z \|D_{c,Z}\psi\| + M_Z\|\psi\|
  \end{equation}
  and there exists a $\mu\in(0,1)$ -- independent of $Z$ -- such that
  \begin{equation}
    \label{eq:epsilon-M}
   \epsilon_{Z} +\frac{M_{Z}}{\sqrt{1-\kappa^2} c^2} \leq \mu.
  \end{equation}
  Then $A\in \tilde\gA$.
\end{lemma}
\begin{proof}
  Note that necessarily $\epsilon_{Z}<1$, i.e., $D_{c,Z}+A(Z)$
  defined as the operator sum on $\gD(D_{c,Z})=\gD(D_{c,Z}+A(Z))$ is
  self-adjoint.

  Since $|D_{c,Z}|\geq c^2(1-\kappa^2)^\frac12$, Inequality
  \eqref{eq:form} implies
  \begin{equation}
      |A(Z)|\leq \left(\eps_Z+\frac{M_Z}{\sqrt{1-\kappa^2}c^2} \right)|D_{c,Z}|\lesssim |D_{c,0}|
    \end{equation}
    which yields \eqref{eq:Birman}.
  
    To show \eqref{eq:ope-A} we pick $\psi\in\gD$,  and estimate
    by Hardy's inequality
\begin{equation}
  \|D_{c,Z}^A\psi\|\leq \|D_{c,Z} \psi\|+\|A(Z) u\|
                      \leq  \left(1+\epsilon_{Z,A(Z)}+ {M_{Z}\over\sqrt{c^4-c^2Z^2}}\right)\|D_{c,Z} \psi\|_{\gH}.
  \end{equation}
  Thus, the upper bound in \eqref{eq:ope-A} is true with
  $$C:=(1+2\kappa)(1+\mu)\geq (1+2\kappa)\left(1+\epsilon_{Z}+
    M_{Z}/\sqrt{c^4-c^2Z^2}\right).$$

  Lower bound:
  \begin{equation}\label{eq:2.12}
    C^{-1}|D_{c,0}|\leq |D_{c,Z}^{A(Z)}|.
\end{equation}
We set
\begin{equation}
  \label{signum}
    \sgn(t):=\begin{cases}
        1, \qquad &t\geq 0;\\
        -1,\qquad & t<0.
    \end{cases}
\end{equation}
Proceeding as in Fournais et al. \cite[Below
(31)]{Fournaisetal2020}, we can get
\begin{equation}
  \begin{split}
  &|D_{c,0}|^{\frac12}{\frac{1}{|D_{c,Z}^{A(Z)}|^{\frac12} }} = |D_{c,0}|^{\frac12}{\frac{1}{D_{c,Z}^{A(Z)}}}|D_{c,Z}^{A(Z)}|^{\frac12} U_Z\\
  =&\Big( |D_{c,0}|^{{\frac12}}{\frac{1}{D_{c,Z}}}|D_{c,Z}^{A(Z)}|^{\frac12}\\
  &\phantom{\Big(}-|D_{c,0}|^{\frac12} {\frac{1}{D_{c,Z}} }\sqrt{|A(Z)|}{\frac{1}{ \sgn(A(Z))+B_Z}}\sqrt{|A(Z)|}
  {\frac{1}{D_{c,Z}}}|D_{c,Z}^{A(Z)}|^{\frac12}\Big) U_Z
  \end{split}
\end{equation}
with $U_Z=\sgn(D_{c,Z}^{A(Z)})$ and a Birman-Schwinger kernel
$B_Z:=\sqrt{|A(Z)|} D_{c,Z}^{-1}\sqrt{|A(Z)|}$.

From \eqref{eq:form} we have
\begin{equation}
  \label{eq:2.38}
    |A(Z)|\leq \left(\epsilon_Z +\frac{M_Z}{\sqrt{1-\kappa^2} c^2} \right)|D_{c,Z}|. 
\end{equation}
Therefore
\begin{equation}
  \|B_Z\|\leq \||A(Z)|^{\frac12}|D_{c,Z}|^{-\frac12}\|^2\leq \epsilon_Z +\frac{M_Z}{\sqrt{1-\kappa^2} c^2} \leq\mu<1
\end{equation}
and
\begin{equation}
    \left\| \frac{1}{\sgn(A(Z))+B_Z}\right\|\leq \frac{1}{1-\mu}<\infty.
\end{equation}
By Lemma \ref{lem:DF-ope} and \eqref{eq:2.12} we have 
\begin{equation}
  \begin{split}
  &\left\||D_{c,0}|^{\frac12}{\frac{1}{|D_{c,Z}^{A(Z)}|^{\frac12}}}\right\|\\
  \leq&\left\| |D_{c,0}|^{\frac12}{\frac{1}{|D_{c,Z}|^{\frac12}} } \right\|
   \left\|{\frac{1}{|D_{c,Z}^{A(Z)}|^{\frac12}}}|D_{c,Z}|^{\frac12}\right\|
  \left(1+  \frac{1}{1- \mu} \left\|\sqrt{|A(Z)|} {\frac{1}{|D_{c,Z}|^{\frac12}}}\right\|^2\right)\\
   \leq & C_\kappa^{1/2} \sqrt{1+\mu} \left(1- \mu\right)^{-1}.
  \end{split}
\end{equation}
where the last step uses \eqref{eq:epsilon-M}. Note that this bound is
uniform in $Z$.  Thus we get
\begin{equation}
   C^{-1} |D_{c,0}|\leq |D_{c,Z}^{A(Z)}|
\end{equation}
which is the lower bound.

Eventually we turn to \eqref{eq:trace*op'}. This, however, follows from
\eqref{eq:2.38} and \eqref{eq:epsilon-M}.
\end{proof}

An alternative criterion is 
\begin{lemma}\label{lem:ope-A-s'}
  Suppose $A(Z)$ is a family of symmetric operators in $\gH$ with $Z\geq Z_0$ and for any such $Z$ there exist constants
  $\epsilon_Z, M_Z\in [0,\infty)$ such that for all $\psi\in \gQ$
  \begin{equation}
    \label{eq:form'}
   \|A(Z)\psi\| \leq c\epsilon_{Z} \||\bp|\psi\| + M_{Z}\|\psi\|
  \end{equation}
  and there exist $\mu,\mu'\in(0,1)$ -- independent of $Z$ -- such that
  \begin{equation}
    \label{eq:epsilon-M'}
  C_\kappa^{-1}( \epsilon_{Z} +c^{-2}M_{Z}) \leq \mu.
  \end{equation}
   Then $A\in \gA$.
\end{lemma}
\begin{proof}
By  Lemma \ref{lem:DF-ope} (with $\rho=0$) and \eqref{eq:form'}, we know that
\begin{equation}
  \begin{split}
    &|A(Z)|\leq \sqrt{2}c\epsilon_Z |\bp|+ \sqrt{2}M_Z
    \leq \sqrt{2}\epsilon_Z|D_{c,0}|+ \sqrt{2}M_Z\\
    \leq &\sqrt{2} C_\kappa^{-1}\epsilon_Z |D_{c,Z}|+\sqrt{2}M_Z\leq \mu' |D_{c,Z}|+\sqrt{2}M_Z.
  \end{split}
\end{equation}
Then the selfadjoint extension in the sense Nenciu \cite{Nenciu1976},
is just give as the operator sum. Analogously to the proof Lemma
\ref{lem:ope-A-s}, \eqref{eq:Birman} is also fulfilled.

Concerning the upper bound of \eqref{eq:ope-A}, we have that
\begin{equation}
    \|D_{c,Z}^{A(Z)} \psi\|\leq \|D_{c,Z}\psi\|+\|A\psi\|\leq \left(1+2\kappa+\epsilon_{Z}+c^{-2}M_{Z}\right)\|D_{c,0}\psi\|.
\end{equation}
Thus, the upper bound in \eqref{eq:ope-A} is true with
  $C=(1+2\kappa+\mu).$

Note that from Lemma \ref{lem:DF-ope}, \eqref{eq:form'} and $c^2\leq |D_{c,0}|$, we know 
\begin{equation}\label{eq:2.38'}
    |A(Z)| \leq (\epsilon_Z+c^{-2}M_Z)|D_{c,0}|\leq C_\kappa^{-1}(\epsilon_Z+c^{-2}M_Z)|D_{c,Z}|.
\end{equation}
Then, proceeding as for Lemma \ref{lem:ope-A-s}, we know that if 
\begin{equation}
    C_\kappa^{-1}(\epsilon_Z+c^{-2}M_Z)\leq \mu<1,
\end{equation}
we have
\begin{equation}
    \left\||D_{c,0}|^{1/2}|D_{c,Z}^{A(Z)}|^{-\frac12}\right\|\leq C_\kappa^{1/2}\sqrt{1+\mu}(1-\mu)^{-1}.
\end{equation}
This is the lower bound. Finally, \eqref{eq:trace*op'} follows
immediately from \eqref{eq:epsilon-M'} and \eqref{eq:2.38'}.
\end{proof}

Next we give a sufficient criterion for the last admissibility
condition, i.e., for \eqref{eq:trace*op}.
\begin{lemma}
  \label{sufficient}
  If $A\in\gA$ is a selfadjoint family of operators such that
  \begin{equation}
    \label{eq:sufficient}
    \tr\left(|A(Z)|^\frac12\gamma_*|A(Z)|^\frac12\right)\left\||A(Z)|^\frac12|D_{c,0}|^{-1}|A(Z)|^\frac12\right\|=O(Z^\frac{12}5),\ (Z\to\infty),
  \end{equation}
  then $A$ fulfills \eqref{eq:trace*op}
\end{lemma}
\begin{proof}
  We write $A=|A|^\frac12 \sgn(A) |A|^\frac12$ with $\sgn(s)$ as in
  \eqref{signum} and use the H\"older inequality for trace ideals.
\end{proof}

  \subsection{Examples}
  \label{examples}
We will discuss some examples of admissible potentials.
\subsubsection{Fuzzy pictures including the Furry
  picture}\label{ex:Fuzzy}

Suppose $\gamma_Z$ is a family of signed measure $-1\leq \gamma_Z\leq1$  with $|\gamma_Z|\in\gX$ for all $Z$. Then we call 
\begin{equation}
  \label{hart}
  \phi(Z):= \rho_{\gamma_Z}*|\cdot|^{-1}
\end{equation}
the family of Hartree potential associated to $\gamma$ and the family
of operators $X$ with the integral kernel
  \begin{equation}
  \label{ex}
  [X(Z)](x,y)={\gamma_Z(x,y)\over|x-y|}
\end{equation}
the family exchange operators associated to $\gamma$. 

The following results shows the admissibility of these mean field
potentials under modest hypotheses.
\begin{lemma}
  \label{lem:Hartree-Exchange}
  Pick a family of signed density matrices $\gamma(Z)$, i.e.,
  $|\gamma(Z)|\in\gX$,
such that
  \begin{equation}
    \label{eq:inf}
    \tr [(|D_{c,0}|-c^2)|\gamma| ]= O(Z^\alpha)
  \end{equation}
  and 
  \begin{equation}
    \int_{\rz^3}\rho_{|\gamma|}= O (Z^\beta)
  \end{equation}
  with $\alpha+\beta<4$, $9\alpha+13\beta< 194/5$, $\alpha-\beta\leq2$,
  $\alpha,\beta\geq0$. 
  Then $\phi$, $X$, $\phi-X$ are admissible. 
\end{lemma}
We note that a physical reasonable assumption like $\alpha=7/3$ and
$\beta=1$ is covered easily.

Before embarking on the proof we remind the reader of the Lane-Emden
inequality (see \cite[(44)]{Fournaisetal2020} and
\cite[(1)]{Carlenetal2024}: For any $\gamma\in\gX$, by the
Hardy-Littlewood-Sobolev inequality and the Lieb–Thirring inequality
\begin{equation}\label{eq:D[gamma]}
  \begin{split}
    &2\cD[\rho_\gamma]\leq C^{\rm HLS}\|\rho_{\gamma}\|_{L^{6/5}(\R^3)}^2\\
    \leq &C^{\rm HLS}\left(\tr \gamma\right)^{2/3}\int_{\R^3}\rho_{\gamma}^{4/3}\leq \frac{C^{\rm HLS}}{C^{\rm D}} (\tr\gamma)^\frac23\tr[|\bp|\gamma].
  \end{split}
\end{equation}

\begin{proof}[Proof of Lemma \ref{lem:Hartree-Exchange}:]
  We use Lemma \ref{lem:ope-A-s} to show that $\phi\in \gA$. We start
  with \eqref{eq:form}: Without loss of generality we can assume that
  $\phi\geq0$: if we can show \eqref{eq:epsilon-M} for $\phi_\pm$ with
  $\mu<1/2$, then this implies that $\phi_++\phi_-$ fulfills it with
  $\mu<1$. Moreover, if $A$ and $B$ fulfill \eqref{eq:trace*op} then
  also $a A + b B$, $a,b\in\rz$ fulfills it.

  Since \eqref{eq:form} is homogeneous in $\psi$ it is sufficient to
  verify it for normalized states $\psi$. Using
  \cite[(2.13)]{Mengetal2024} we get
  \begin{equation}
    \label{65}
    \begin{split}
      &\|\rho_{\gamma_Z}*|\cdot|^{-1}\psi\|
        \leq C\tr(|\bp|\gamma_Z)\\
      \leq&
            {\tr\left(\left((|D_{c,0}|-c^2)+\lambda^2c^2\right)\gamma_Z\right)\over\lambda Z} 
            \leq C {Z^\alpha +\lambda^2Z^{2+\beta}\over \lambda Z}
            = CZ^{\frac12(\alpha+\beta)}
    \end{split}
  \end{equation}
  with $\lambda= Z^{\frac12(\alpha-\beta)-1}\leq1$ which yields
  \eqref{eq:form} with $\epsilon_{Z}=CZ^{\frac12(\alpha+\beta)-2}$
  which tends to zero as $Z\to\infty$. Moreover, it implies
  $M(Z)=o(Z^2)$. Thus, for $Z$ large enough \eqref{eq:epsilon-M} holds
  for any positive $\mu$, in particular also, if $\mu$ is smaller than
  $1/2$.

  To verify the remaining condition \eqref{eq:trace*op} we use Lemma
  \ref{sufficient}: We first focus on the left side of
  \eqref{eq:sufficient}. By the Schwarz inequality,
  \eqref{eq:D[gamma]}, and similarly to \eqref{65}
  \begin{equation}
    \begin{split}
    \label{66}
    &0\leq \tr(|\phi|^\frac12\gamma_*|\phi|^\frac12) \leq \cD(\rho_*,\rho_{\gamma_Z})
    \leq \cD[\rho_*]^\frac12 \cD[\rho_{\gamma_Z}]^\frac12 \leq
    CZ^{\frac76+\frac\beta3} \sqrt{\tr (|\bp||\gamma_Z|)}\\
    \leq& C Z^{\frac76+\frac\beta3+\frac14(\alpha+\beta)}= CZ^{\frac76+\frac14\alpha+\frac7{12}\beta}
  \end{split}
\end{equation}
  The second factor is estimated as follows
  \begin{equation}
    \label{67}
    \|\phi(Z)^\frac12|D_{c,0}|^{-1}\phi(Z)^\frac12\|\leq C Z^{\frac12(\alpha+\beta)-2}  \end{equation}
  Thus,
  \begin{equation}
    \label{68}
    \tr(|\phi(Z)|^\frac12\gamma_*|\phi(Z)|^\frac12)\||\phi(Z)^\frac12|D_{c,0}|^{-1}|\phi(Z)^\frac12\|\leq CZ^{-\frac56+\frac34\alpha+\frac{13}{12}\beta}=o(Z^\frac{12}5)
  \end{equation}
using the hypothesis $9\alpha+13\beta<194/5$ in th last step.

The admissibility $X$ is a consequence of the above proof for $\phi$,
since $X_{|\gamma_Z|}\leq \rho_{|\gamma_Z|}*|\cdot|^{-1}$ in the
operator sense, and since
$\|X_\gamma\| \leq \|\rho_{|\gamma_Z|}*|\cdot|^{-1}\|$.

The difference of $\phi-X$ follows by same reason that allowed us to
restrict the proof to non-negative $\phi$.
\end{proof}

We remark that the choice $\gamma_Z=\gamma_*$ fulfills the assumptions
of Lemma \ref{lem:Hartree-Exchange}. Thus the Furry picture is also
covered by our assumptions.

\subsubsection{Weird pictures including the free picture} \phantom{x}
\label{ex:exterior}

\textbf{Coulomb potentials}: Note that
$D_{c,Z}^{(Z-Z')/|\cdot|}= D_{c,Z'}$. The operator can defined in the
sense Nenciu for $-1<\kappa':=Z'/c<1$.
\begin{itemize}
\item If $Z'=0$, then the projection becomes
  $\Lambda^+=\1_{(0,\infty)}(D_{c,0})$ yielding the free picture,
  i.e., the Brown--Ravenhall operator.
\item If $Z'=Z$, then the projection is $P_Z$ it yields the Furry
  picture.
\item In all other cases the projections are given by  hydrogenic Dirac
  operators with nuclear charge $Z'$.
\end{itemize}
Thus, all the above coupling parameters $Z'$ except for $Z'=Z$ are
``weird'', i.e., we cannot expect the correct physical ground state
energy in leading and subleading order but they are nevertheless
covered by \eqref{eq:Mittleman} under certain conditions on $Z'$:

\begin{lemma}
  \label{noname}
  Assume $\kappa'\in(-1,\kappa)$ and
  $A(Z):= (1-\kappa'/\kappa) Z |\cdot|^{-1}$. Then $A$ is admissible.
\end{lemma}
\begin{proof}
  From \eqref{eq:D-kappa-form} we know that \eqref{eq:ope-A} is
  satisfied, if $\kappa'\in (-1,1)$.

  Next we address \eqref{eq:trace*op'}: By \eqref{eq:D-kappa-form} and
  Kato's inequality, for any $\kappa'\leq\kappa$ there exists a
  constant $C_{\kappa'}$ such that
\begin{equation}
    \frac{1}{|x|}\leq \frac{\pi}{2c}|D_{c,0}|\leq {\pi\over2c C_{\kappa'}}|D_{c,(\kappa'/\kappa) Z}|.
\end{equation}
This yields
\begin{equation}
  -A(Z)\leq (\kappa'/\kappa-1)_+Z|\cdot|^{-1}
  \leq {\pi\over2 C_{\kappa'}}(\kappa'-\kappa)_+|D_{c,(\kappa'/\kappa)Z}|.
\end{equation}
Thus \eqref{eq:trace*op'} is true, if
$ 2\pi^{-1} C_{\kappa'}>(\kappa'-\kappa)_+$ or
\begin{equation}
  \label{79} \kappa'<\kappa+ 2C_{\kappa'}/\pi.
\end{equation}

Combining the conditions yields that $ A \in\gA$, if $-1<\kappa'\leq\kappa$.
\end{proof}

We remark that the last estimate of the proof can be improved, since
$C_\kappa$ is monotone decreasing in $\kappa$. Together with the
explicit formula \eqref{5} this allows for a simple numerical
improvement as indicated in Figure \ref{Bild}:
\begin{figure}[h]
    \includegraphics[width=7cm]{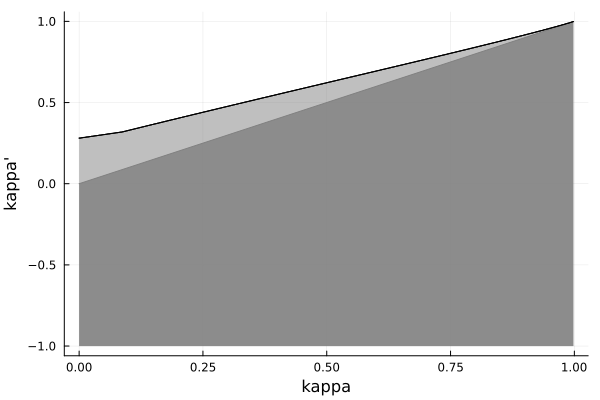}
    \caption{{\small Numerically admissible $(\kappa,\kappa')$: The dark
      shaded part is the set of the hypothesis of Lemma
      \ref{noname}. The light shaded region indicates the additional
      points from numerical evaluation of \eqref{79} using
      \ref{5} in the  Definition of $C_{\kappa'}$.}}\label{Bild}
\end{figure}

\section{Asymptotic behavior of heavy atoms}\label{sec:asymptotic}
In this section, we are going to prove Theorem \ref{th:mittleman}. The
proof is based on Theorem \ref{th:DF-Z7/3}. We start by recalling some
well know identities:

\subsection{Some identities}

For any symmetric operator $B$ such that $D_{c,Z}+B$ is selfadjoint
and $0\not\in \sigma(D_{c,Z}+B)$ we have
\begin{equation}\label{eq:projector}
  \1_{(0,\infty)}(D_{c,Z}+B)=\frac{1}{2}+\frac{1}{2\pi}\int_{-\infty}^{\infty} \frac{1}{D_{c,Z}+B+\ri z}\rd z.
\end{equation}
Thus, since $0\not\in \sigma(D_{c,Z}+{A(Z)})$ for any $A\in \gA$ by
\eqref{eq:ope-A}, we have
\begin{equation}
P_{Z,A(Z)}= \frac{1}{2}+\frac{1}{2\pi}\int_{-\infty}^{\infty} \frac{1}{D_{c,Z}+A(Z)+\ri z}\rd z.
\end{equation}
From now on -- in abuse of notation -- we will also write $A$ for
$A(Z)$.  In particular, we have
\begin{equation}
    P_*= \frac{1}{2}+\frac{1}{2\pi}\int_{-\infty}^{\infty} \frac{1}{D_*+\ri z}\rd z.
\end{equation}
Then
\begin{equation}\label{eq:PA-P1A} P_{Z,A}=P_*+ Q_{Z,A} \end{equation}
with
\begin{equation}
  \begin{split}
    &Q_{Z,B}=-\frac{1}{2\pi} \int_{-\infty}^{\infty} \frac{1}{D_*+\ri z}\underbrace{(A- \phi_*)}_{=:A_*} \frac{1}{D_{c,Z}+B+\ri z}\rd z\\
    &=-\frac{1}{2\pi} \int_{-\infty}^{\infty} \frac{1}{D_{c,Z}+B+\ri z} A_*\frac{1}{D_*+\ri z}\rd z. 
  \end{split}
\end{equation}
Moreover
\begin{equation}\label{eq:P1A-P2A}
    Q_{Z,A}= Q_{Z,\phi_*}+ R_{Z,A}
\end{equation}
where
\begin{equation}
    R_{Z,B}=\frac{1}{2\pi}\int_{-\infty}^{\infty}\frac{1}{D_*+\ri z}A_*\frac{1}{D_{c,Z}+B+\ri z}A_*\frac{1}{D_*+\ri z} \rd z.
\end{equation}

Thus, for any eigenspinor $\psi_j$ of $D_*$
\begin{equation}
  \begin{split}
     &Q_{Z,B}\psi_j=-\frac{1}{2\pi}\int_{-\infty}^{\infty}\frac{1}{D_{c,Z}+B+\ri z}\frac{1}{\lambda_j+\ri z}\rd z A_*\psi_j\\
    = &\frac{1}{2\pi} \frac{1}{D_{c,Z}+B -\lambda_j}\mathcal{P}_{Z,B,j}^\perp \int_{-\infty}^{\infty}\left( \frac{1}{D_{c,Z}+B+\ri z} -\frac{1}{\lambda_j+\ri z}\right) \rd z A_*\psi_j\\
    = &\frac{1}{2} \frac{1}{D_{c,Z}+B-\lambda_j}  \left(\sgn(D_{c,Z}+B)-1\right)A_*\psi_j\\
  = &- \frac{1}{D_{c,Z}+B  -\lambda_j} P^\perp_{Z,B}
     A_*\psi_j=  - \frac{1}{D_{c,Z}+B  -\lambda_j} P^\perp_{Z,B} A_*\psi_j.
  \end{split}
\end{equation}
where $\mathcal{P}_{Z,B,j}^\perp= \mathbbm{1}_{\R\setminus\{\lambda_j\}}(D_{c,Z}+B)$.

So
\begin{align}\label{eq:P1-gamma*}
    Q_{Z,B}\gamma_*= -\sum_{j}\mu_j \frac{1}{D_{c,Z}+B -\lambda_j} P^\perp_{Z,B} A_*\left|\psi_j\right>\left<\psi_j\right|
\end{align}
where $\mu_1,\mu_2,...$ are the eigenvalues of $\gamma_*$ associated to the corresponding eigenfunctions $\psi_1,\psi_2,...$.
and
\begin{equation}\label{eq:P1-gamma*-P1}
  Q_{Z,B}\gamma_*Q_{Z,B}
     = \sum_{j} \frac{\mu_j}{D_{c,Z}+B -\lambda_j} P^\perp_{Z,B} A_*\left|\psi_j\right>\left<\psi_j\right|A_* P_{Z,B}^{\perp}\frac{1}{D_{c,Z}+B -\lambda_j}.
 \end{equation}

\subsection{Proof of Theorem \ref{th:mittleman}}\label{sec:mittleman}

That the Mittleman energy is at least the Furry energy follows from
\cite{HandrekSiedentop2015}. This is the reason, why the theorem is
stated as an upper bound to which we turn now:

\begin{proof}
We can estimate the energy from above by restricting to Slater
determinants. Then proceeding like Lieb \cite{Lieb1981V} (see also
Bach \cite{Bach1992} for an alternative proof) and dropping the exchange
term we get for any $\psi\in \gQ^N$ 
\begin{equation}
  \label{90}
    \cE_{Z}[\psi] \leq \cE_{Z}^\mathrm{H}(\gamma_\psi)
  \end{equation}
  where $\gamma_\psi$ is the one-particle density matrix of the
  fermionic state $\psi$.  Thus, to prove our claim we merely need to
  pick an appropriate trial density matrix $\gamma$ and evaluate the
  right-hand side of \eqref{90}. To this end we set
  \begin{align}\label{eq:Gamma-A}
     \Gamma_{Z}^A:=\{\gamma\in \Gamma_{Z}\;;\; P_{Z,A} \gamma P_{Z,A}=\gamma\}.
  \end{align}

We take $\gamma_A=P_{Z,A}\gamma_* P_{Z,A}\in \Gamma_{Z,N}^A$. Then we know that
\begin{equation}\label{eq:EMN=>EMrhf}
     E_{Z,A}=\inf\{\cE[\psi]|\psi\in\gQ^N_A,\ \|\psi\|\leq1\}\leq \cE_{Z}^\mathrm{H}(\gamma_A).
\end{equation}
By \eqref{eq:PA-P1A} and \eqref{eq:P1A-P2A},
\begin{equation}
  \begin{split}
    \label{eq:93}
    &\gamma_A=(P_*+Q_{Z,A})\gamma_*(P_*+Q_{Z,A})=\gamma_*+2\Re(Q_{Z,A}\gamma_*)
      +Q_{Z,A}\gamma_* Q_{Z,A}\\
    =&\gamma_*+ 2\Re(Q_{Z,\phi_*}\gamma_*)+2\Re(R_{Z,A}\gamma_*)
       + Q_{Z,A}\gamma_* Q_{Z,A}.
  \end{split}
\end{equation}
Thus
\begin{equation}\label{eq:E-DF-decomp}
  \begin{split}
    &\cE^\mathrm{H}_{Z}(\gamma_A)=\cE_{Z}^\mathrm{H}(\gamma_*)+\tr[(D_*-c^2) (\gamma_A-\gamma_*)]+\cD[\rho_{\gamma_A}-\rho_*]\\
    =&\underbrace{\cE_{Z}^\mathrm{H}(\gamma_*)+\tr[(D_*-c^2) (Q_{Z,\phi_*}\gamma_*+\gamma_* Q_{Z,\phi_*})]}_{=:I}\\
    &+\underbrace{\tr[(D_*-c^2) (2\Re(R_{Z,A}\gamma_*)+ Q_{Z,A}\gamma_* Q_{Z,A})]}_{=:II} \\
    &+\underbrace{\cD[\rho_{2\Re(Q_{Z,A}\gamma_*)} +\rho_{Q_{Z,A}\gamma_*Q_{Z,A}}]}_{=:III}.
  \end{split}
\end{equation}
To continue the proof, we are going to bound the right-hand side of \eqref{eq:E-DF-decomp} term by term.

\medskip

{\bf Estimate on $I$.} By \eqref{eq:P1-gamma*} the operator $Q_{Z,\phi_*}\gamma_*$ is off-diagonal, i.e.,
\begin{equation}
  \label{offdiagonal}
  P_*^\perp  Q_{Z,\phi_*}\gamma_* P_*=  Q_{Z,\phi_*}\gamma_*.
\end{equation}
By \eqref{offdiagonal}, its adjoint, and $Q_{Z,\phi_*}^*=Q_{Z,\phi_*}$ we have
\begin{equation}
  \label{eq:96}
  \begin{split}
  &\tr[(D_*-c^2) (Q_{Z,\phi_*}\gamma_*+\gamma_* Q_{Z,\phi_*})]\\
    = &\tr[(D_*-c^2) ( P_*^\perp  Q_{Z,\phi_*}\gamma_* P_* + P_*  \gamma_* Q_{Z,\phi_*} P_*^\perp)] =0.
  \end{split}
\end{equation}
Thus, by Theorem \ref{th:DF-asymptotic} and \eqref{eq:Z7/3-1} 
\begin{equation}\label{eq:I}
    I=\cE_{Z}^{\rm H}(\gamma_*)=E_Z^S+\tfrac1ND[\rho_*] =E_Z^F+o(Z^2).
\end{equation}

\medskip

{\bf Estimate on $II$.} By \eqref{eq:P1-gamma*-P1} we can estimate the
last term in $II$ as follows:
\begin{equation}\label{eq:P-gamma-P-A}
  \begin{split}
    &\tr[(D_*-c^2)Q_{Z,A}\gamma_* Q_{Z,A}]\\
    = &\sum_{j}\mu_j\left<\psi_j,A_*P_{Z,A}^{\perp}\frac{1}{\lambda_j-D_{c,Z}-A} (D_*-c^2)\frac{1}{\lambda_j-D_{c,Z}-A} P_{Z,A}^{\perp}  A_*\psi_j\right>\\
    \leq & \sum_{j}\mu_j\left<\psi_j,A_*P_{Z,A}^{\perp}\frac{1}{\lambda_j-D_{c,Z}-A} \phi_*\frac{1}{\lambda_j-D_{c,Z}-A} P_{Z,A}^{\perp}  A_*\psi_j\right>\\
     \leq & \sum_{j}\mu_j\left<\psi_j,A_*P_{Z,A}^{\perp}\frac{1}{\lambda_j+|D_{c,Z}+A|} \phi_*\frac{1}{\lambda_j+|D_{c,Z}+A|} P_{Z,A}^{\perp}  A_*\psi_j\right>\\
    \lesssim &Z^{\frac43-2} \sum_{j}\mu_j\left<\psi_j,A_*\frac{1}{\lambda_j+|D_{c,Z}+A|}   A_*\psi_j\right>\\
    \lesssim &Z^{-\frac23} \tr(A_*\frac{1}{|D_{c,0}|} A_*\gamma_*)
    = Z^{-\frac23} \tr((A-\phi_*)\frac{1}{|D_{c,0}|} (A-\phi_*)\gamma_* =O(Z^\frac{26}{15})
  \end{split}
\end{equation}
where we used \eqref{eq:trace*op'} for the first inequality,
\eqref{eq:ope-A} and \eqref{eq:Z7/3-1} for the second inequality, and
the Schwarz inequality in the last step to estimate the mixed terms,
and the last admissibility condition \eqref{eq:trace*op} which we can
apply to $A$ but also to $\phi_*$, since it is admissible as well (see
Lemma \ref{lem:Hartree-Exchange} with $\alpha=7/3$ and $\beta=1$ and
Theorem \ref{th:DF-Z7/3}). (We remind the reader that we suppressed
the $Z$ dependence of $A$ in this section.)

Next we estimate the term $\tr[(D_*-c^2) R_{Z,A}\gamma_*]$ by
\eqref{eq:gamma}:
\begin{equation}
  \label{eq:99}
  \begin{split}
    &\Re(\tr[(D_*-c^2) R_{Z,A}\gamma_*])\\
    =&\sum_j{\mu_j\over2\pi}\int_{-\infty}^{\infty}\Re\left<(D_*-c^2)\psi_j, \frac{1}{D_*+\ri z}A_* \frac{1}{D_{c,Z}+A+\ri z}A_*\frac{1}{D_*+\ri z}\psi_j\right>\rd z \\
    =& -\frac{1}{2\pi} \Re\sum_{j}\mu_j \int_{-\infty}^{\infty} \frac{\lambda_j-c^2}{(\lambda_j+\ri z)^2} \left<\psi_j, A_* \frac{1}{D_{c,Z}+A+\ri z} A_* \psi_j \right> \rd z\\
    =& \sum_{j}\mu_j  \left<\psi_j, A_*  P_{Z,A}^{\perp}\frac{\lambda_j-c^2}{(\lambda_j-D_{c,Z}-A)^2} P_{Z,A}^{\perp} A_* \psi_j \right> <0.
  \end{split}
\end{equation}
Thus for any $A\in\cA$
\begin{align}\label{eq:II}
  II \leq 0.
\end{align}

\medskip

{\bf Estimate on $III$.} It remains to study $III$. The Schwarz
inequality for the summation implies
\begin{equation}
    |\rho_{\Re(Q_{Z,A}\gamma_*)}(x)|\leq [\rho_{Q_{Z,A}\gamma_* Q_{Z,A}}(x)]^{1/2}[\rho_{\gamma_*}(x)]^{1/2}.
  \end{equation}
  By the triangular inequality for the Coulomb norm
  $\|\cdot\|_C:=\sqrt{\cD[\cdot]}$ we have
\begin{equation}
  \sqrt{III} \leq 2\|\rho_{\Re(Q_{Z,A}\gamma_*)}\|_C
  +\|\rho_{Q_{Z,A}\gamma_*Q_{Z,A}}\|_C
\end{equation}
and again by the Schwarz inequality -- now in the Coulomb scalar product --
\begin{equation}
 \left|\cD[\rho_{Q_{Z,A}\gamma_*} ]\right|\leq   \cD[\rho_{Q_{Z,A}\gamma_* Q_{Z,A}}]^{1/2}\cD[\rho_*]^{1/2}.
\end{equation}
Thus
\begin{equation}\label{eq:III-decomp}
  III\lesssim \cD[\rho_{Q_{Z,A}\gamma_* Q_{Z,A}}]^\frac12  \cD[\rho_*]^\frac12
  + \cD[\rho_{Q_{Z,A}\gamma_* Q_{Z,A}}]
\end{equation}

We will use \eqref{eq:D[gamma]} to estimate
$\cD[\rho_{Q_{Z,A}\gamma_* Q_{Z,A}}]$. Since its proof uses the
Lieb-Thirring inequality we have to make sure that
$0\leq Q_{Z,A}\gamma_* Q_{Z,A}\leq1$ and if not, we would have to
divide by its norm and treat that operator by \eqref{eq:D[gamma]}. To
this end pick $u,v\in \gH$, set
$B_A:=|D_{c,0}|^{-\frac12}A|D_{c,0}|^{-\frac12}$, and compute
\begin{multline}
  \left|\left<v,Q_{Z,A}u\right>\right|
  \leq\frac{1}{2\pi}\int_{\R}\left|\left<v, \frac{1}{D_{c,Z}+A+\ri \eta} A_*  \frac{1}{D_*+\ri \eta} u\right>\right|\rd \eta\\
  \leq \frac{1 }{2\pi}\|B_{A_*}\|
  \sqrt{\int_{\R}\left\||D_{c,0}|^\frac12 {1\over D_{c,Z}+A+\ri \eta}v\right\|^2\rd \eta\int_{\rz}\left\||D_{c,0}|^\frac12  \frac{1}{D_*+\ri \eta} u\right\|^2\rd \eta}\\
  \lesssim \frac{1 }{2\pi}\|B_{A_*}\|
  \sqrt{
    \int_{\R}\left\| {|D_{c,Z}+A|^\frac12\over D_{c,Z}+A+\ri \eta} v\right\|^2\rd \eta
    \int_{\rz}\left\|{|D_*|^\frac12  \over D_*+\ri \eta} u\right\|^2\rd \eta}
  =\frac12 \|B_{A_*}\| =O(1)
\end{multline}
where we used \eqref{eq:Birman} in the last step. Thus
\begin{equation}
    \|Q_{Z,A}\gamma_*Q_{Z,A}\|\leq\|Q_{Z,A}\|^2\|\gamma_*\|=O(1),
  \end{equation}
  i.e., it is a density up to a constant uniform in $Z$ which is
  irrelevant for our purposes and allows us to apply
  \eqref{eq:D[gamma]} directly up to a multiplicative constant uniform
  in $Z$. We get
  \begin{equation}
    \begin{split}
  &2\cD[\rho_{Q_{Z,A}\gamma_*Q_{Z,A}}]\lesssim (\tr (Q_{Z,A}\gamma_*Q_{Z,A}))^{2/3}\int \rho^{4/3}_{ Q_{Z,A}\gamma_* Q_{Z,A}}\\
  \lesssim& (\tr (Q_{Z,A}\gamma_*Q_{Z,A}))^{2/3} \tr (|\bp|Q_{Z,A}\gamma_* Q_{Z,A})\\
  \lesssim& (\tr (Q_{Z,A}\gamma_*Q_{Z,A}))^\frac23  \tr (|\bp|Q_{Z,A}\gamma_* Q_{Z,A}).
    \end{split}
  \end{equation}

Taking the trace in \eqref{eq:P1-gamma*-P1} yields
\begin{equation}\label{eq:3.15}
  \begin{split}
  &\tr[ Q_{Z,A}\gamma_* Q_{Z,A}]=\sum_{j} \mu_j \left<\psi_j, A_* P_{Z,A}^{\perp}\frac{1}{(D_{c,Z}+A-\lambda_j)^2} P_{Z,A}^{\perp} A_* \psi_j\right>\\
  \leq &\tr\left(A_*
    \frac{1}{(|D_{c,Z}+A|+\lambda_1)^2}A_*\gamma_*\right) \lesssim
  c^{-2}\tr\left(A_*{1\over |D_{c,0}|} A_*\gamma_*\right)\\
   \leq &2c^{-2}\tr\left(A{1\over |D_{c,0}|} A\gamma_*\right)+2 c^{-2}\tr\left(\phi_*{1\over |D_{c,0}|} \phi_*\gamma_*\right)=o(Z^\frac25)
  \end{split}
\end{equation}
where we were estimating as in \eqref{eq:P-gamma-P-A}.

Analogously using $|\bp|\leq c^{-1}|D_{c,Z}+A|$ by
\eqref{eq:ope-A}, we have
\begin{equation}
  \begin{split}
  &\tr(|\bp| Q_{Z,A}\gamma_* Q_{Z,A})\\
    \leq &\sum_{j} \mu_j \left\langle\psi,A_*\frac{1}{|D_{c,Z}+A|+\lambda_j}
      |\bp|\frac{1}{|D_{c,Z}+A|+\lambda_j} A_* \psi_j\right)\\
    \lesssim &c^{-1}\sum_{j} \mu_j \left\langle\psi,A_*\frac{1}{|D_{c,Z}+A|+\lambda_j}
      |D_{c,Z}+A|\frac{1}{|D_{c,Z}+A|+\lambda_j} A_* \psi_j\right)\\
    \leq&
    c^{-1}\tr\left(A_*{1\over|D_{c,Z}+A|+\lambda_1}A_*\gamma_*\right)\\
    \lesssim& c^{-1}\tr\left(A_*{1\over|D_{c,Z}+A|+\lambda_1}A_*\gamma_*\right)
     =o(Z^{\frac75}).
  \end{split}
\end{equation}
Thus, by the above we get
\begin{equation}
  \cD[\rho_{Q_{Z,A}\gamma_* Q_{Z,A}}] = o(Z^{\frac23\cdot\frac25+\frac75})=o(Z^\frac56)
\end{equation}
and therefore
\begin{equation}\label{eq:III}
    III = o(Z^{\frac56\cdot\frac12+\frac76})=o(Z^2).
\end{equation}

{\bf Conclusion.} Now we can conclude from \eqref{eq:EMN=>EMrhf},
\eqref{eq:E-DF-decomp}, \eqref{eq:I}, \eqref{eq:II}, and
\eqref{eq:III}
\begin{equation}
  E_{Z}^{A}\leq \cE^{\rm H}_{Z}(\gamma_{A(Z)})+o(Z^2)
  =C^\mathrm{TF}Z^{\frac73}+C_\mathrm{Scott}^{\rm F}Z^2+o(Z^2).
\end{equation}
This proves Theorem \ref{th:mittleman}.
\end{proof}

\appendix
\section{Existence of S\'er\'e minimizers for large $Z$}
\label{sec:existence}

In this appendix, we are going indicate the prove of Theorem
\ref{th:DF-existence}. The proof is essentially the same as in
\cite{Sere2023} except for some modifications based on
\cite{Fournaisetal2020} and \cite{Meng2023}. We set 
\begin{equation}
    \rho_{N,\gamma}=\tfrac{N-1}{N}\rho_{\gamma}
\end{equation}
in this section.

Note that $\Gamma_{Z,N}$ is not a convex set which complicates the
existence proof of minimizers. To overcome this problem, S\'er\'e  introduced the retraction mapping
\begin{equation}
    \theta_Z(\gamma):=\lim_{n\rightarrow \infty} T_Z^n(\gamma)
\end{equation}
with
\begin{equation}
  T_Z^n(\gamma)=T_Z(T_Z^{n-1}(\gamma)),\quad T_Z(\gamma)=P_{Z, \rho_{N,\gamma}}\gamma P_{Z, \rho_{N,\gamma}}.
\end{equation}
The existence of $\theta_Z$ is proven in \cite[Proposition
2.9]{Sere2023}. However, instead of S\'er\'e's original version, we
use the variant of Fournais et al.  \cite{Fournaisetal2020} which better
fits our goal to cover the entire range of kappas.
\begin{lemma}\label{lem:retra}
  Assume that $N\leq Z$, $M>0$, $Z$ be large enough such that
  $L:=c^{-1/6}C_{\kappa}^{\rm ret}M<1$ (with $c=\kappa^{-1}Z$),
  $\lambda>(1-L)^{-1}$, and set
\begin{equation}
    \cU_{Z,N,M}=\{\gamma\in \Gamma_{N},\; 
   c^{-2} \|\gamma\|_{X_c}+c^{-2} \lambda\|T_Z(\gamma)-\gamma\|_{X_c}\leq MN\}.
\end{equation}
Then $T_Z$ maps $ \cU_{Z,N,M}$ into $ \cU_{Z,N,M}$, and for any $\gamma\in \cU_{Z,N,M}$ the sequence $(T_Z^n(\gamma))_{n\geq 0}$ converges to a limit $\theta_Z(\gamma)\in\Gamma_{Z,N}^+$. Moreover for any $\gamma\in \cU_{Z,N,M}$,
\begin{equation}\label{eq:retra}
    \begin{split}
    \|T_Z^{n+1}(\gamma)-T_Z^{n}(\gamma)\|_{X_c}\leq &L \|T_Z^n(\gamma)-T_Z^{n-1}(\gamma)\|_{X_c},\\
    \|\theta_Z(\gamma)-T_Z^n(\gamma)\|_{X_c}\leq&\frac{L^n}{1-L}\|T_Z(\gamma)-\gamma\|_{X_c}.
\end{split}
\end{equation}
\end{lemma}
\begin{proof}
  The existence of the retraction can be found in \cite[Pages
  593-594]{Fournaisetal2020} modulo the invariance condition
  $U_a \gamma U_a^{-1}=\gamma$ for any $a\in SU(2)$ and the
  Fermi-Amaldi correction. Using the notations of
  \cite{Fournaisetal2020}, we set $U_{Z,N,M}:=F$, $L:=k$, and
  $C_{\kappa}^{\rm ret}:=C_{\kappa,\nu}$ ($\nu=\kappa$) with $F$, $k$,
  and $C_{\kappa,\nu}$ being defined by \cite[pages 593-594 and
  Eq. (55)]{Fournaisetal2020} for any $\kappa<\nu_0$ with $\nu_0$.

  We can extend the condition $\kappa<\nu_0$ to $\kappa<1$ by
  replacing \cite[Theorem 4]{Fournaisetal2020} by Lemma
  \ref{lem:DF-ope}. In fact for any $\kappa\in (0,1)$ we have
  $C_{\kappa}^{\rm ret}<\infty$, if $U_a \gamma U_a^{-1}=\gamma$ for
  any $a\in SU(2)$.

  To conclude it remains to show that $T_Z$, and therefore also $\theta_Z$,
  are invariant under the representation $U_a$, i.e., for any
  $\gamma\in \Gamma_{N}$ and any $a\in SU(2)$
\begin{equation}
    U_aT_Z(\gamma) U_a^{-1}=T_Z( U_a \gamma  U_a^{-1})=T_Z(\gamma)\ \ \text{and}\ \
     U_a\theta_Z(\gamma) U_a^{-1}=\theta(\gamma).
\end{equation}
Thus, by the invariance of $\gamma$,
$\rho_{\gamma}(x)= \rho_{\gamma}(R_a^{-1}x)$, i.e.,
$\rho_{\gamma}$ is spherically invariant. By \eqref{eq:A-RA-identity},
\begin{equation}\label{eq:UA-commute}
    [D_{c,Z, \rho_{N,\gamma}},U_a]=[P_{Z, \rho_{N,\gamma}},U_a]=0.
\end{equation}
As a result,
\begin{equation}
     U_aT_Z(\gamma) U_a^{-1}= U_a (P_{Z, \rho_{N,\gamma}} \gamma P_{Z, \rho_{N,\gamma}}) U_a^{-1} =T_Z( U_a \gamma  U_a^{-1})=T_Z(\gamma)
\end{equation}
and by iteration
\begin{equation}
    U_aT_Z^n(\gamma) U_a^{-1}=T_Z^n( U_a \gamma  U_a^{-1})=T_Z^n(\gamma).
\end{equation}
Taking the limit yields
\begin{equation}
    U_a\theta_Z(\gamma) U_a^{-1}=\theta_Z(\gamma).
\end{equation}
\end{proof}

Lemma \ref{lem:retra} is the analogue of \cite[Lemma
2.3]{Meng2023}. It allows us to prove an adapted version of
\cite[Proposition 3.2]{Meng2023}
\begin{lemma} \label{approx}
  Assume $L$, $N,Z,M$ be given as in Lemma \ref{lem:retra},
  $\gamma\in \cU_{Z,N,M}\cap \Gamma_{Z,N}$, and $h\in\gX$ such that
  $P_{Z, \rho_{N,\gamma}}h P_{Z, \rho_{N,\gamma}}=h$. Then for any
  $t\in\mathbb{R}$ satisfying $\gamma+th\in \cU_{Z,N,M}$ and $Z$ large
  enough, we have
  \begin{equation}\label{eq:sed-diff}
    \begin{split}
      &\cE_Z^{\rm S}(\theta_Z(\gamma+th))\\
      =&\cE_Z^{\rm S}(\gamma)+t\;\tr [(D_{c,Z,\rho_{N,\gamma}}-c^2)h]+\tfrac12\tfrac{N-1}N t^2\tr[W_h h]+\frac{t^2}{c^2} {\rm Err}_{\gamma}(h,t),
    \end{split}
  \end{equation}
where 
\begin{align}
     |{\rm Err}_{\gamma}(h,t)|\leq \frac{5\pi^2(6+\pi)(1+M)N}{4C_{\kappa}^4(1-\kappa)^{5/2}(1-L)^2}\|h\|_X^2.
\end{align}
\end{lemma}
\begin{proof}
  We want to roll back the proof to the proof of \cite[Proposition
  3.2]{Meng2023}. To this end we set there $\alpha=1-1/N$ and
  $R=MN$. Since the proof of \cite[Proposition 3.2]{Meng2023} only
  depends on \eqref{eq:retra} and Lemma \ref{lem:DF-ope}, we may
  replace \cite[Lemma 2.3]{Meng2023} by Lemma \ref{lem:retra} without
  changing the result. In addition, we replace the estimate
\begin{equation}
    |D_{c,0}|\leq (1-\kappa')^{-1}|D_{c,Z,\rho_{\gamma}}|
  \end{equation}
  by Lemma \ref{lem:DF-ope} which amounts to replacing the factor
  $(1-\kappa')^{-4}$ by $C_{\kappa}^{-4}$. This now holds for all
  $\kappa<1$. Finally, we also use that in our situation
\begin{equation}
    \lambda_0:=1-c^{-1}\max\{N,Z\}=1-\kappa.
\end{equation}
with $\lambda_0$ being defined by \cite[Assumption 2.2]{Meng2023}.
\end{proof}
Lemmata \ref{lem:retra} and \ref{approx} replace
\cite[Theorem 2.10]{Sere2023}.

We also need the following estimates for the minimizing sequence which
will be used to replace \cite[Proposition 2.11]{Sere2023} and
\cite[Lemma 3.2]{Sere2023} respectively.
\begin{lemma}\label{lem:bound-DF}
  Assume that $N\leq Z$. Let $\gamma\in \Gamma_{Z,N}$ be such that
  $\cE_Z^{\rm S}(\gamma)\leq 0$.  Then for $Z\geq 8\kappa^3\left(
    \frac{ C^{\rm HLS}}{C_\kappa' C^{\rm D}}\right)^3$, we have
\begin{equation}
    \|\gamma\|_{X_c} \leq 2(C_\kappa')^{-1}Nc^2.
\end{equation}
\end{lemma}
\begin{proof}
According to the proof of \cite[Proposition 2.11]{Sere2023}, for any $\gamma\in \Gamma_{Z,N}$ such that $ \cE_Z^{\rm S}(\gamma)\leq 0$, we have
\begin{equation}
0\geq   \cE_{Z,N}^{\rm HA}(\gamma)=\tr\big(|D_{c,Z,\rho_{N,\gamma}}|\gamma\big) -\tfrac{N-1}{N}\cD[\gamma]-c^2N.
\end{equation}
By \eqref{eq:D[gamma]}
\begin{equation}
    2\cD[\gamma]\leq  \frac{C^{\rm HLS}}{C^{\rm D}}
    N^{2/3}\tr[|\bp|\gamma]\leq  \frac{C^{\rm HLS}}{C^{\rm D}}
    N^{2/3}c^{-1} \|\gamma\|_{X_c}.
\end{equation}
Using Lemma \ref{lem:DF-ope} gives
\begin{equation}
\left( C_\kappa'-\frac{C^{\rm HLS}}{C^{\rm D}}c^{-1}N^{2/3} \right)\|\gamma\|_{X_c}\leq c^2 N.
\end{equation}
Choosing
$Z\geq 8\kappa^3\left( \frac{ C^{\rm HLS}}{C_\kappa' C^{\rm
      D}}\right)^3$ yields
\begin{equation}
    \frac{C^{\rm HLS}}{C^{\rm D}}c^{-1}N^{2/3}\leq \frac{C^{\rm HLS}}{C^{\rm D}}c^{-1}Z^{2/3}\leq\frac{1}{2}C_\kappa'.
\end{equation}
Thus we get $\|\gamma\|_{X_c} \leq 2(C_\kappa')^{-1}Nc^2$.
\end{proof}

Analogously to \cite[Corollary 2.12]{Sere2023}, we have the
following.
\begin{corollary}
Assume that 
  \begin{align}
      Z> C_\kappa^{\rm ex}=8\kappa\max\left\{\kappa^2\left({ C^{\rm HLS}\over C_\kappa' C^{\rm D}}\right)^3, 8\left({ C_{\kappa}^{\rm ret}\over C_\kappa'}\right)^6\right\}.
  \end{align}
  Then one can choose $M=2(C_\kappa')^{-1}$ and $\rho>0$ such that for
  all $\gamma \in \Gamma_{Z,N}$ satisfying
  $ \cE_Z^{\rm S}(\gamma)\leq 0$, there holds
  $B_X(\gamma,\rho)\cap \Gamma_{N}\subset \cU_{Z,N,M}$.
\end{corollary}
\begin{proof}
   For any $\gamma \in \Gamma_{Z,N}$ satisfying $ \cE_Z^{\rm S}(\gamma)\leq 0$, we have
\begin{equation}
    T_Z(\gamma)=\gamma.
\end{equation}
By Lemma \ref{lem:bound-DF}, for
$Z\geq 8\kappa^3(\frac{C^{\rm HLS}C_{\rm LT}}{C_\kappa'} )^3$,
 \begin{equation}
     c^{-2} \|\gamma\|_{X_c}+c^{-2} \lambda\|T_Z(\gamma)-\gamma\|_{X_c}=  c^{-2} \|\gamma\|_{X_c}\leq 2(C_\kappa')^{-1}.
 \end{equation}
Choosing $M=2(C_\kappa')^{-1}$, then if $Z> \kappa\left(2C_{\kappa}^{\rm ret}/C_\kappa'\right)^6$, we have $L=C_{\kappa}^{\rm ret}c^{-1/6} M<1$.
\end{proof}

Using Lemma \ref{lem:DF-ope}, we can replace \cite[Lemma
3.2]{Sere2023} by the following.
\begin{corollary}
  Let $\widetilde{\gamma},\gamma\in \Gamma_{N}$ with
  $\gamma\leq\mathbbm{1}_{[0,\nu]}(D_{c,Z,\rho_{\widetilde{\gamma}}})$
  for some $\nu>0$. Then 
    \begin{equation}
      \|\gamma\|_{X_c}\leq (C_\kappa')^{-1}\nu N.
    \end{equation}
\end{corollary}

Finally using the above lemmata allows us to replace \cite[Lemma
3.3]{Sere2023} by
\begin{lemma}
  Assume that $N\leq Z$ and let $\gamma_n$ be a minimizing sequence in
  $\Gamma_{Z,N}^+$.  Then for $Z$ large enough,
    \begin{equation}
        \lim_{n\to \infty}\tr[(D_{c,Z,\rho_{N,\gamma_n}}-c^2)\gamma_n]-\inf_{\substack{\gamma\in \Gamma_{N}\\ P_{Z,\rho_{N,\gamma_n}}\gamma P_{Z,\rho_{N,\gamma_n}} =\gamma}}\tr[(D_{c,Z,\rho_{N,\gamma}}-c^2)\gamma]\to 0.
    \end{equation}
\end{lemma}
\begin{proof}
  This proof follows exactly the one of \cite[Lemma 3.3]{Sere2023}
  except we replace \cite[Theorem 2.10]{Sere2023} by Lemma
  \ref{approx}. To end the proof, it suffices to show that there
  exists a sequence $(g_n)$ of bounded self-adjoint operators of rank
  $N$ in $\Gamma_N$ such that
  $0\leq g_n \leq \mathbbm{1}_{[0,\nu]}(D_{c,Z,\rho_{N,\gamma_n}})$
  and
\begin{equation}
  \tr\left((D_{c,Z,\rho_{N,\gamma_n}}-c^2)g_n\right)\leq \inf_{\substack{\gamma\in \Gamma_{N}\\ P_{Z,\rho_{N,\gamma_n}}\gamma P_{Z,\rho_{N,\gamma_n}} =\gamma}}\tr[(D_{c,Z,\rho_{N,\gamma_n}}-c^2)\gamma]+\frac{\epsilon_0}{2}.
\end{equation}
First of all, it is easy to see that there exists a sequence
$(\widetilde{g}_n)$ of bounded self-adjoint operators of rank $N$ in
$\Gamma_N$ such that
$0\leq \widetilde{g}_n \leq
\mathbbm{1}_{[0,\nu]}(D_{c,Z,\rho_{N,\gamma_n}})$,
$\tr(\widetilde{g}_n)\leq N$ and
\begin{equation}
    \tr\left((D_{c,Z,\rho_{N,\gamma_n}}-c^2)\widetilde{g}_n\right)\leq \inf_{\substack{\gamma\in\gX\\\tr(\gamma)\leq N \\ P_{Z,\rho_{N,\gamma_n}}\gamma P_{Z,\rho_{N,\gamma_n}} =\gamma}}\tr[(D_{c,Z,\rho_{N,\gamma_n}}-c^2)\gamma]+\frac{\epsilon_0}{2}.
\end{equation}
Proceeding as for \eqref{eq:UA-commute}, we get for any $a\in SU(2)$
that $[D_{c,Z,\rho_{N,\gamma_n}}, U_a]=0$.  So for any $a\in SU(2)$,
\begin{multline}
   \tr\left((D_{c,Z,\rho_{N,\gamma_n}}-c^2)U_a \widetilde{g}_n U_a^{-1}\right)
  = \tr\left((D_{c,Z,\rho_{N,\gamma_n}}-c^2)\widetilde{g}_n\right)\\
    \leq \inf\left\{\tr[(D_{c,Z,\rho_{N,\gamma_n}}-c^2)\gamma]\big|
           \gamma\in\gX,\ \tr(\gamma)\leq N,\ P_{Z,\rho_{N,\gamma_n}}\gamma P_{Z,\rho_{N,\gamma_n}} =\gamma\right\}
    +\frac{\epsilon_0}{2}.
\end{multline}
Finally, set
\begin{equation}
    g_n:=\int_{SU(2)}U_a \widetilde{g}_n U_a^{-1} \rd\mu_H(a). 
\end{equation}
where $\mu_H$ is the Haar measure on $SU(2)$.

Then as $\widetilde{g}_n\in\gX$ and $\tr(\widetilde{g}_n)\leq N$, we
have that $g_n \in \Gamma_{N}$ and
\begin{equation}
    \|g_n\|_{X_c}\leq \|\widetilde{g}_n\|_{X_c}\leq C_{\kappa}^{-1}\nu N
\end{equation}
is bounded. Proceeding as in \cite[Lemma 3.3]{Sere2023} using Lemma
\ref{approx} proves the lemma.
\end{proof}

Finally, we replace \cite[Lemma 3.6]{Sere2023} by the following one.
\begin{lemma}\label{lem:spectral}
    Assume that $1\leq N\leq Z$. Then
    \begin{itemize}
    \item $\sigma_\mathrm{ess}(D_{c,Z, \rho_{N,\gamma_n}})=(-\infty,-c^2]\cup[c^2,\infty)$.
    \item
      $\forall_{\gamma\in \Gamma_N}
      \lambda_n(D_{c,Z})\leq \lambda_n(D_{c,Z,\rho_{N,\gamma}})\leq \lambda_n(D_{c,1})$ counting multiplicity.
          \item
      $\sigma_d(D_{c,Z, \rho_{N,\gamma}})\subset
      I:=[c^2\sqrt{1-\kappa^2},c^2)$ and
      $|\sigma_d(D_{c,Z, \rho_{N,\gamma_n}})|=\infty$.
     \item For $\gamma\in \Gamma_N$ and any $e\in I$
       $$\mathrm{dim}(\1_{[0,e]}(D_{c,Z,\rho_{N,\gamma}}))
       \leq \mathrm{dim}(\1_{[0,e]}(D_{c,Z})).
      $$
    \end{itemize}
\end{lemma}
\begin{proof}
  Since $D_{c,Z,\rho_{N,\gamma}}^{-1}-D_{c,0}^{-1}$ is compact, the
  first claim follows from Weyl's theorem (see, e.g., Reed and Simon
  \cite[Theorem XIII.14]{ReedSimon1978}).

  The second item follows from
  $D_{c,Z}\leq D_{c,Z,\rho_{N,\gamma}}\leq D_{c,1}$ and perturbation
  theory.

  The remaining claims follow from this fact, and the fact that
  $|\sigma_d(D_{c,\zeta})|=\infty$ (Gordon \cite{Gordon1928}, Griesmer
  and Lutgen \cite[Theorem 3 (i)]{GriesemerLutgen1999}) for $c>\zeta>0$.
\end{proof}
Eventually, repeating the remaining steps of the proof in
\cite{Sere2023} (for the case $N<Z$, in our case even for $N\leq Z$),
we get Theorem \ref{th:DF-existence} including the Euler equation
\begin{align}\label{eq:scf}
    \gamma_*=\mathbbm{1}_{(0,\mu)}(D_*)+\delta
\end{align}
with $0\leq \delta \leq \mathbbm{1}_{\{\mu\}}(D_*)$ and
$\mu\in (c^2\sqrt{1-\kappa^2}, c^2)$. Then Theorem \ref{th:property}
follows directly from Lemma \ref{lem:spectral} and \eqref{eq:scf}.

\section{Properties of S\'er\'e minimizers}\label{sec:Z7/3}
We begin with a proof of Theorem \ref{th:DF-Z7/3}.  Let $\gamma_{*}$
be a minimizer of $E^\mathrm{S}_{Z}$ in $\Gamma_{Z,Z}$, and let
$Z':=Z+Z^{2/3}$. The main argument is to show that
  \begin{equation}
    \begin{split}
    &\frac{E_{Z}^{\rm S}-E_{Z'}^{\rm S}}{Z^{2/3}}\geq \frac{\cE_{Z}^{\rm S}(\gamma_*)-\cE_{Z'}^{\rm S}\Big(\theta_{Z'}(\gamma_*)\Big)}{Z^{2/3}}\\
    =& \frac{\cE_{Z}^{\rm S}(\gamma_*)-\cE_{Z'}^{\rm HA}(\gamma_*)}{Z^{2/3}}+ \frac{\cE_{Z'}^{\rm HA}(\gamma_*)-\cE_{Z'}^{\rm S}\Big(\theta_{Z'}(\gamma_*)\Big)}{Z^{2/3}} \geq C \tr[|\cdot|^{-1}\gamma_*].
    \end{split}
  \end{equation}

  Note that $E_{Z}^{\rm S}$ is also dependent on $\kappa$. So we write
  $E_Z^{\rm S}(\kappa)$ to emphasis the dependence on $\kappa$. From
  Theorem \ref{th:DF-asymptotic}, we know that for any
  $0<\eps\leq 1-\kappa$ there exists two constants $C'_\mathrm{Scott}$ such
  that
\begin{align*}
     \inf_{\kappa'\in [\kappa,\kappa+\epsilon]}  E_Z^{\rm S}(\kappa')\geq C^\mathrm{TF}Z^{7/3}+C'_\mathrm{Scott}Z^2
\end{align*}
Since $C^\mathrm{TF}<0$, for $Z$ large enough, we get
\begin{equation}
\begin{split}
    & \frac{E_{Z}^{\rm S}(\kappa)-E_{Z'}^{\rm S}(\kappa+c^{-1}Z^{2/3})}{Z^{\frac23}}\leq  \frac{E_{Z}^{\rm S}(\kappa)-\inf_{\kappa'\in [\kappa,\kappa+\epsilon]}E_{Z'}^{\rm S}(\kappa') }{Z^{\frac23}} \\
     \leq&  \frac{C^\mathrm{TF}Z^{7/3}(1-(1+Z^{-1/3})^{7/3})+CZ^2-C'_\mathrm{Scott}(Z')^2}{Z^{\frac23}}= O(Z^{4/3}).
\end{split}
\end{equation}

To do so, we first prove that $\theta_{Z'}(\gamma_*)$ is well-defined
in $\Gamma_{Z'}$. For further convenience, we set
\begin{equation}
   \widetilde{D}_{\frac23}:=D_{c,Z',\frac{Z'-1}{Z'}\rho_{\gamma_*}},\qquad  \widetilde{P}_{\frac23}:=\mathbbm{1}_{(0,\infty)}(\widetilde{D}_{\frac23}).
\end{equation}

\subsection{Existence of $\theta_{Z'}(\gamma_*)$} We are now going to
show that $\gamma_*\in \cU_{Z',Z',M}$ 
for some $M>0$, so
$\theta_{Z'}(\gamma_*)$ is well-defined in $\Gamma_{Z',Z'}$.

First of all, it is easy to see that $\gamma_*\in \Gamma_{Z'}$. Then
we are going to find a constant $M>0$ independent of $Z$ such that
\begin{equation}
     c^{-2} \|\gamma_*\|_{\gX_c}+c^{-2} \lambda\|T_{Z'}(\gamma_*)-\gamma_*\|_{\gX_c}\leq MN.
\end{equation}
According to Lemma \ref{lem:bound-DF}, we have
\begin{equation}
    c^{-2}\|\gamma_*\|_{\gX_c}\leq 2(C_\kappa')^{-1}N.
\end{equation}
In addition, note that $P_*\gamma_* P_*=\gamma_*$. Thus
\begin{equation}
    T_{Z'}(\gamma_*)-\gamma_*= (\widetilde{P}_{\frac23}-P_*)\gamma_*\widetilde{P}_{\frac23}+\gamma_*(\widetilde{P}_{\frac23}-P_*).
\end{equation}
By \eqref{eq:P1-gamma*}
\begin{equation}
  (\widetilde{P}_{\frac23}-P_*)\gamma_*= -\sum_{j}\mu_j \frac{1}{\widetilde{D}_{\frac23} \widetilde{P}_{\frac23}^\perp -\lambda_j} \widetilde{P}_{\frac23}^\perp \widetilde{A}(Z)\left|\psi_j\right>\left<\psi_j\right|
\end{equation}
with
$\widetilde{A}(Z)=
-Z^{\frac23}|\cdot|^{-1}+\frac{Z^{\frac23}}{ZZ'}\rho_{*}*|\cdot|^{-1}$.
By Newton's lemma,
\begin{align*}
     |\widetilde{A}(Z)|\leq 2Z^{\frac23}|\cdot|^{-1}\lesssim Z^{-1/3}|D_{c,0}|.
\end{align*}
According to Lemma \ref{lem:DF-ope} with $Z$ being replaced by
$Z'$,
\begin{equation}\label{eq:4.1}
  \begin{split}
 &   \left\||D_{c,0}|^{\frac12}\widetilde{P}_{\frac23}^\perp (\widetilde{D}_{\frac23} \widetilde{P}_{\frac23}^\perp -\lambda_j)^{-\frac12} \widetilde{P}_{\frac23}^\perp |\widetilde{A}(Z)|^{\frac12} \right\|\\
    \lesssim &Z^{-\frac16}\left\||D_{c,0}|^{\frac12}\widetilde{P}_{\frac23}^\perp
               (|\widetilde{D}_{\frac23}| \widetilde{P}_{\frac23}^\perp +\lambda_j)^{-\frac12}\right\|^2\\
    \lesssim&_\kappa Z^{-\frac16}\left\||\widetilde{D}_{\frac23}|^{\frac12}\widetilde{P}_{\frac23}^\perp (|\widetilde{D}_{\frac23}| \widetilde{P}_{\frac23}^\perp +\lambda_j)^{-\frac12}\right\|^2 \lesssim_\kappa Z^{-\frac16}.
  \end{split}
\end{equation}
Using Lemma \ref{lem:bound-DF},
\begin{equation}\label{eq:4.2}
    \||D_{c,0}|^{1/2}\widetilde{P}_{\frac23} |D_{c,0}|^{-1/2}\|\lesssim_\kappa \|\widetilde{P}_{\frac23}|\widetilde{D}_{\frac23}|^{1/2} |D_{c,0}|^{-1/2}\|\lesssim_\kappa 1,
\end{equation}
we have
\begin{equation}\label{eq:T2/3-gamma}
  \begin{split}
    &\| T_{Z'}(\gamma_{*})-\gamma_{*}\|_{\gX_c}\lesssim_\kappa Z^{-1/6}\||\widetilde{A}(Z)|^{1/2} \gamma_*|D_{c,0}|^{1/2}\|_1\\
    \lesssim& Z^{-1/3}\|\gamma_*\|_{\gX_c}\lesssim_\kappa Z^{-1/3} c^2N =O(Z^{8/3}).
  \end{split}
\end{equation}
Thus there exists a constant $M>0$ such that
\begin{equation}
    c^{-2}\|\gamma_*\|_{\gX_c}+c^{-2}\lambda \|T_{Z'}(\gamma)-\gamma\|_{\gX_c}\leq MZ.
\end{equation}
This shows that $\gamma_*\in \cU_{Z',M}$. Then according to Lemma
\ref{lem:retra}, for $Z$ large enough,
$\theta_{Z'}(\gamma_*)\in\gX_c\cap \Gamma_{Z',Z'}^+$.

\subsection{Estimates on $|\cE_{Z'}^{\rm S}(\theta_{Z'}(\gamma_*))-\cE_{Z'}^{\rm HA}(\gamma_*)|$}
In this subsection, we are going to use \cite[Lemma 5.1]{Meng2023}
with $R=MZ$ which is
\begin{lemma}\label{lem:error-bound}
For any $\gamma\in\mathcal{U}_{Z',Z',M}$ and $Z$ large enough, 
\begin{equation}
    |\cE_{Z'}^{\rm S}(\theta_{Z'}(\gamma))-\mathcal{E}_{Z'}^{\rm HA}(\gamma)|\lesssim_\kappa c^{-3}\|T_{Z'}(\gamma)-\gamma\|_{\gX_c}^2+\|\widetilde{P}_{\frac23}^\perp\gamma \widetilde{P}_{\frac23}^\perp\|_{\gX_c}
\end{equation}
 \end{lemma}
By Newton's lemma,
\begin{equation}
   - Z^{\frac23}|\cdot|^{-1} \leq \widetilde{A}(Z)\leq -Z^{\frac23} (1- (Z')^{-1})|\cdot|^{-1}.
\end{equation}
Then according to Lemma \ref{lem:bound-DF}, \eqref{eq:T2/3-gamma},
\begin{equation}
  \begin{split}
    &\|T_{Z'}(\gamma_*)-\gamma_*\|_{\gX_c}^2 =Z^{-1/3}\||\widetilde{A}(Z)|^{1/2}\gamma_* |D_{c,0}|^{1/2}\|_1^2\\
    \lesssim& Z^{-1/3}Z^{3}\tr[|\widetilde{A}(Z)|\gamma_*]=Z^{10/3}\tr[|\cdot|^{-1}\gamma_*].
  \end{split}
\end{equation}
It remains to study $\|\widetilde{P}_{\frac23}^\perp\gamma_* \widetilde{P}_{\frac23}^\perp\|_{\gX_c}$. By $P_*\gamma_*P_*=\gamma_*$, we get
\begin{equation}
    \widetilde{P}_{\frac23}^\perp\gamma_* \widetilde{P}_{\frac23}^\perp= \widetilde{P}_{\frac23}^\perp(\widetilde{P}_{\frac23}-P_*)\gamma_* (\widetilde{P}_{\frac23}-P_*) \widetilde{P}_{\frac23}^\perp.
\end{equation}
Note that
\begin{equation}
  \begin{split}
 & (\widetilde{P}_{\frac23}-P_*)\gamma_* (\widetilde{P}_{\frac23}-P_*)\\
  =& - \sum_{j}\mu_j \frac{1}{\widetilde{D}_{\frac23} \widetilde{P}_{\frac23}^\perp -\lambda_j} \widetilde{P}_{\frac23}^\perp \widetilde{A}(Z)\left|\psi_j\right>\left<\psi_j\right|\widetilde{A}(Z) \widetilde{P}_{\frac23}^\perp\frac{1}{\widetilde{D}_{\frac23} \widetilde{P}_{\frac23}^\perp -\lambda_j}.
  \end{split}
\end{equation}
Therefore by \eqref{eq:4.1}-\eqref{eq:4.2}
\begin{equation}
  \|\widetilde{P}_{\frac23}^\perp\gamma \widetilde{P}_{\frac23}^\perp\|_{\gX_c}\lesssim_{\kappa} Z^{-1/3}\||\widetilde{A}(Z)|^{1/2}\gamma_* |\widetilde{A}(Z)|^{1/2}\|_1\leq  Z^{1/3}\tr[|\cdot|^{-1}\gamma_*].
\end{equation}
Thus,
\begin{equation}\label{eq:E-theta-2/3}
    |\cE_{Z'}^{\rm S}(\theta_{Z'}(\gamma_*))-\mathcal{E}_{Z'}^{\rm HA}(\gamma_*)|\lesssim_\kappa Z^{1/3}\tr[|\cdot|^{-1}\gamma_*].
\end{equation}

\subsection{Proof of Theorem \ref{th:DF-Z7/3}}

\subsubsection{Proof of Estimates \eqref{eq:Z7/3-1}}
As $Z'=Z+Z^{\frac23}$, it is easy to see that 
\begin{equation}\label{eq:E-2/3-decomp}
  \cE_{Z'}^{\rm HA}(\gamma_*)= \cE_{Z}^{\rm S}(\gamma_*) -Z^{\frac23}\tr[|\cdot|^{-1}\gamma_*].
\end{equation}
According to \eqref{eq:E-theta-2/3} and \eqref{eq:E-2/3-decomp}, for
$Z$ large enough,
\begin{equation}
  \begin{split}
    &  \cE_{Z}^{\rm S}(\gamma_*)-\cE_{Z'}^{\rm S}(\theta_{Z'}(\gamma_*))\\
    \geq & \cE_{Z}^{\rm S}(\gamma_*)-\cE_{Z'}^{\rm HA}(\gamma_*)-\left|\cE_{Z'}^{\rm HA}(\gamma_*) -\cE_{Z'}^{\rm S}(\theta_{Z'}(\gamma_*))\right|\\
    \geq &Z^{\frac23}\tr[|\cdot|^{-1}\gamma_*]-C(\kappa) Z^{1/3}\tr[|\cdot|^{-1}\gamma_*]\geq \tfrac{1}{2}Z^{\frac23}\tr[|\cdot|^{-1}\gamma_*].
  \end{split}
\end{equation}
Since $C^\mathrm{TF}<0$, for $Z$ large enough
\begin{equation}\label{eq:coulomb}
  \tr[|\cdot|^{-1}\gamma_*]\lesssim \frac{E_Z^{\rm S}-E_{Z'}^{\rm S}}{Z^{\frac23}}
  \leq \frac{C^\mathrm{TF}Z^\frac73(1-(1+Z^{-\frac13})^\frac73)+CZ^2}{Z^{\frac23}}= O(Z^\frac43). 
\end{equation}
This gives the first estimate in \eqref{eq:Z7/3-1}.

By Newton's lemma and the fact that $\rho_*$ is spherical invariant,
\begin{equation}
  0\leq\phi_*(x)=\tfrac{Z-1}{Z}\rho_{\gamma_*}*|\cdot|^{-1}(x)\leq \int_{\R^3}\frac{\rho_*}{|\cdot|}=O(Z^{4/3}).
\end{equation}
This proves second estimate in \eqref{eq:Z7/3-1}. 

The last estimate in \eqref{eq:Z7/3-1} follows directly from the
second estimate by using Newton's lemma again.

We remark, though, that the last estimate in \eqref{eq:Z7/3-1} does
not hinge on the spherical symmetry of the density as the following
alternative argument shows: By \eqref{eq:gamma}, we have
\begin{equation}
  \cE_{Z}^{\rm S}(\gamma_*)+\tfrac{Z-1}{Z}\cD[\gamma_*]
  = \tr[(D_{c,Z,\frac{Z-1}{Z}\rho_*}-c^2)\gamma_*]\leq 0;
\end{equation}
by Theorem \ref{th:DF-asymptotic}, we have
\begin{equation}
    \cE_{Z}^{\rm S}(\gamma_*) =E_{Z}^{\rm S} = C^\mathrm{TF}Z^{\frac73} +O(Z^2).
\end{equation}
Thus
\begin{equation}\label{eq:D-gamma}
    \cD[\gamma_*]\leq -\cE_{Z}^\mathrm{S}(\gamma_*)=O(Z^{\frac73}).
\end{equation}
This and \eqref{eq:coulomb} also give the last estimate in \eqref{eq:Z7/3-1}.

\subsubsection{Proof of Estimates \eqref{eq:Z7/3-2}}
By \eqref{eq:Z7/3-1} and Theorem \ref{th:DF-asymptotic} we have
\begin{align}\label{eq:D-c^2}
    |\tr[(D_{c,0}-c^2)\gamma_*]|\leq |E_{Z}^{\rm S}|+\tfrac{Z-1}{Z}\cD[\gamma_*]+ Z\tr[|\cdot|^{-1}\gamma_*]=O(Z^{\frac73}).
\end{align}

We recall that $\Lambda=\mathbbm{1}_{(0,\infty)}(D_{c,0})$. Then
$D_{c,0} \Lambda=|D_{c,0}|\Lambda$ and
$D_{c,0}\Lambda^\perp=-|D_{c,0}|\Lambda^\perp$. Concerning the term
$\tr[(|D_{c,0}|-c^2)\gamma_*]$, we have
\begin{equation}\label{eq:|D_{c,0}|-c^2-decomp}
  \begin{split}
    &\tr[(|D_{c,0}|-c^2)\gamma_*]=   \tr[(|D_{c,0}|-c^2)\Lambda\gamma_* \Lambda]+\tr[(|D_{c,0}|-c^2)\Lambda^\perp\gamma_* \Lambda^\perp]\\
    = &\tr[(D_{c,0}-c^2)\Lambda\gamma_* \Lambda]+\tr[(D_{c,0}-2D_{c,0}-c^2)\Lambda^\perp\gamma_* \Lambda^\perp]\\
       = &\tr[(D_{c,0}-c^2)\gamma_* ]+2\tr[|D_{c,0}|\Lambda^\perp\gamma_* \Lambda^\perp].
  \end{split}
\end{equation}
As $P_*\gamma_* P_*=\gamma_*$, we have
\begin{equation}
  \begin{split}
    &\tr[|D_{c,0}|\Lambda^\perp (P_*-\Lambda)\gamma_* ]=\tr[|D_{c,0}| (P_*-\Lambda)\gamma_* (P_*-\Lambda)\Lambda^\perp]\\
    =& \tr[|D_{c,0}| (\Lambda- P_*)\gamma_* (\Lambda- P_*)\Lambda^\perp].
  \end{split}
\end{equation}
According to \eqref{eq:P1-gamma*} again, we have
\begin{equation}
  \begin{split}
    &(\Lambda- P_*)\gamma_*(\Lambda- P_*) \\
    =& \sum_{j}\mu_j \frac{1}{D_{c,0} -\lambda_j} \Lambda^\perp \widetilde{A}_*(Z)\left|\psi_j\right>\left<\psi_j\right|\widetilde{A}_*(Z) \Lambda^\perp \frac{1}{D_{c,0}  -\lambda_j}
\end{split}
\end{equation}
with
$\widetilde{A}_*(Z)= Z|\cdot|^{-1}-\tfrac{Z-1}{Z}
\rho_**|\cdot|^{-1}$. Since $\|\Lambda^\perp\|=1$,
\begin{equation}
    \left|\tr[ (|D_{c,0}|-c^2)(P_*-\Lambda)\gamma_*]\right|\leq \sum_{j}\mu_j\left\| |D_{c,0}|^{1/2} \frac{1}{D_{c,0} -\lambda_j} \Lambda^\perp \widetilde{A}_*(Z)\psi_j\right\|^2.
\end{equation}
Analogously to \eqref{eq:4.1}, as
$0\leq \widetilde{A}_*(Z) \leq Z|\cdot|^{-1}\lesssim |D_{c,0}|$ and
$|D_{c,0}|-c^2\leq |D_{c,0}|$ , we have
\begin{equation}
    \left\||D_{c,0}|^{1/2} \frac{1}{D_{c,0} -\lambda_j} \Lambda^\perp \widetilde{A}_*^{1/2}(Z)\right\|\lesssim 1.
\end{equation}
Therefore
\begin{equation}
    \left|\tr[ (|D_{c,0}|-c^2)(P_*-\Lambda)\gamma_*]\right|\lesssim  \tr[|\widetilde{A}_*(Z)|\gamma_*]\leq Z\tr[|\cdot|^{-1}\gamma_*]=O(Z^{\frac73}).
\end{equation}
Thus, from \eqref{eq:coulomb}, \eqref{eq:D-c^2}, and
\eqref{eq:|D_{c,0}|-c^2-decomp}
\begin{equation}
    \tr[(|D_{c,0}|-c^2)\gamma_*]=O(Z^{\frac73}).
\end{equation}
This proves the first estimate in \eqref{eq:Z7/3-2}.

Finally by \eqref{65} with $\alpha=7/3$ and $\beta=1$ and
$\lambda=Z^{-1/3}$, we have
\begin{equation}
    \tr[|\bp|\gamma_*]\leq \frac{1}{c\lambda}\left(\tr[(|D_{c,0}|-c^2)\gamma_*]+ c^2\lambda^2 Z\right)\lesssim \frac{ Z^{\frac73}}{c\lambda}+ c
    \lambda Z= O(Z^\frac53).
\end{equation}
This gives the second estimate in \eqref{eq:Z7/3-2}.

\subsubsection{Proof of Estimates \eqref{eq:Z7/3-3}}
In this subsection let $A:=(1-1/Z)\rho^{\rm TF}_Z*|\cdot|^{-1}$ and
$A_*=A-\phi_*$. From \eqref{eq:Z7/3-1} and $\|A\|=O(Z^{4/3})$ we know
\begin{align}\label{eq:177}
    \|A_*\|=O(Z^{4/3}).
\end{align}

Note that
\begin{equation}\label{eq:178}
\begin{split}
    &E_{Z}^{\rm F}+o(Z^2)=E^{\rm S}_Z=\tr[(D_{c,Z}-c^2) \gamma_*]+\tfrac{Z-1}{Z}\cD[\rho_*]\\
=& \tr[(D_{c,Z}-c^2+A) \gamma_*]-\tfrac{Z-1}{Z}\cD[\rho^{\rm TF}]+\tfrac{Z-1}{Z}\cD[\rho_*-\rho^{\rm TF}]\\
=&\tr[(D_{c,Z}-c^2+A) P_{Z,A}\gamma_*P_{Z,A}]-\tfrac{Z-1}{Z}\cD[\rho^{\rm TF}]+\tfrac{Z-1}{Z}\cD[\rho_*-\rho^{\rm TF}] \\
&+ \tr[(D_{c,Z}-c^2+A) (\gamma_*-P_{Z,A}\gamma_*P_{Z,A})]
\end{split}
\end{equation}
As in \cite[Section 3]{HandrekSiedentop2015} and by the minimax
principle (see \cite{GriesemerSiedentop1999,Dolbeaultetal2000O}) we get
\begin{equation}
\begin{split}
  &\tr[(D_{c,Z}-c^2+A) P_{Z,A}\gamma_*P_{Z,A}]-\tfrac{Z-1}{Z}\cD[\rho^{\rm TF}]\\
  \geq& \inf_{\gamma\in \Gamma_Z^A} \tr[(D_{c,Z}-c^2+A) \gamma]-\tfrac{Z-1}{Z}\cD[\rho^{\rm TF}]\\
   \geq& \inf_{\gamma\in \Gamma_Z^0} \tr[(D_{c,Z}-c^2+A) \gamma]-\tfrac{Z-1}{Z}\cD[\rho^{\rm TF}]
  =E_Z^{\rm F}+o(Z^2).
\end{split}
\end{equation}
To end the proof, it remains to show that
\begin{align}\label{eq:179}
  \left|\tr[(D_{c,Z}-c^2+A) (\gamma_*-P_{Z,A}\gamma_*P_{Z,A})]\right|=o(Z^2).
\end{align}
Once it is proven, we will get
\begin{align*}
  E_{Z}^{\rm F}+o(Z^2)\geq  E_{Z}^{\rm F}+ \tfrac{Z-1}{Z}\cD[\rho_*-\rho^{\rm TF}] +o(Z^2).
\end{align*}
Thus,
\begin{equation}
    \cD[\rho_*-\rho^{\rm TF}]=o(Z^2).
\end{equation}

Now we prove \eqref{eq:179}. By \eqref{eq:93},
\begin{equation}\label{eq:182}
    \begin{split}
      &\tr[(D_{c,Z}-c^2+A) (\gamma_*-P_{Z,A}\gamma_*P_{Z,A})]\\
      =&- \tr[(D_{c,Z}-c^2+A)(Q_{Z,\phi_*}\gamma_*+\gamma_*Q_{Z,\phi_*})]\\
      &-\tr[(D_{c,Z}-c^2+A) (2\Re (R_{Z,A}\gamma_*)+Q_{Z,A}\gamma_*Q_{Z,A})]
    \end{split}
\end{equation}
From \eqref{eq:P1-gamma*}, \eqref{eq:96}, and \eqref{eq:177}, we know
\begin{equation}\label{eq:183}
    \begin{split}
      &\left|\tr[(D_{c,Z}-c^2+A) (Q_{Z,\phi_*}\gamma_*+\gamma_*Q_{Z,\phi_*})]\right|\\
      =&\left|\tr[A_* (Q_{Z,\phi_*}\gamma_*+\gamma_*Q_{Z,\phi_*})]\right|\\
      =&2\left|\sum_{j=1}\mu_j\Re\left(\left<\psi_j, A_*  \frac{1}{D_*-\lambda_j}P_*^\perp A_* \psi_j\right>\right)\right|\\
      \lesssim& c^{-2}\|A_*\|^2\tr\gamma_*=O(Z^\frac53).
    \end{split}
\end{equation}
where we used $c^2 \lesssim |D_*|$ from
\eqref{2.15}.

Concerning the last term in \eqref{eq:178}, using
\eqref{eq:P1-gamma*-P1}, we get
\begin{equation}\label{eq:184}
    \begin{split}
      &\left|\tr[(D_{c,Z}-c^2+A) Q_{Z,A}\gamma_* Q_{Z,A}]\right|\\
      \leq & \sum_{j}\mu_j\left|\left<\psi_j,A_*P_{Z,A}^{\perp}\frac{D_{c,Z}+A-c^2}{(\lambda_j-D_{c,Z}-A)^2}  P_{Z,A}^{\perp}  A_*\psi_j\right>\right|\\
      \lesssim& c^{-2}\|A_*\|^2 \tr\gamma_* =O(Z^\frac53).
    \end{split}
\end{equation}
Note that
\begin{equation}\label{eq:185}
    \tr[(D_{c,Z}-c^2+A) R_{Z,A}\gamma_*]= \tr[(D_*-c^2) R_{Z,A}\gamma_*]+ \tr[A_* R_{Z,A}\gamma_*].
\end{equation}
From \eqref{eq:99} we obtain
\begin{equation}\label{eq:186}
\begin{split}
  & \left|\tr[(D_*-c^2) R_{Z,A}\gamma_*]\right|\\
  \leq& \sum_{j}\mu_j  \left|\left<\psi_j, A_*  P_{Z,A}^{\perp}\frac{(\lambda_j-c^2)}{(\lambda_j-D_{c,Z}-A)^2} P_{Z,A}^{\perp} A_* \psi_j \right>\right| =O(Z^\frac53).
\end{split}
\end{equation}
Eventually
\begin{equation}\label{eq:187}
   \begin{split}
     &\left|\tr[A_*  R_{Z,A}\gamma_*]\right|\\
     \leq & \frac{1}{2\pi}\sum_j \mu_j \left|\int_{-\infty}^{\infty}\left<A_*\psi_j,\frac{1}{D_*+\ri z}A_*\frac{1}{D_{c,Z}+A+\ri z}A_*\frac{1}{\lambda_j+\ri z} \psi_j\right> \rd z \right|\\
     \leq &\sum_j\mu_j \frac{\|A_*\|\|(D_{c,Z}+A)^{-1}\|}{2\sqrt{\lambda_j} \sqrt{\pi}}\left(\int_{-\infty}^\infty \left\|(D_*+\ri z)^{-1} A_*\psi_j\right\|^2 \rd z\right)^{1/2}\|A_*\psi_j\|\\
     \lesssim & c^{-3}  \sum_j \mu_j\|A_*\|\||D_*|^{-1/2}A_*\psi_j\|\|A_*\psi_j\|\lesssim c^{-4}\|A_*\|^3\tr\gamma_*= O(Z).
   \end{split}
\end{equation}
Combining \eqref{eq:182}-\eqref{eq:187} shows \eqref{eq:179} 
completing the proof of \eqref{eq:Z7/3-3}.

\bigskip

{\bf Acknowledgments.}: Partial support by the Deutsche
Forschungsgemeinschaft (DFG, German Research Foundation) through TRR
352 -- Project 470903074 is acknowledged.\medskip

{\bf Declarations.}\medskip

{\bf Conflict of interest} The authors declare no conflict of interest.\medskip

{\bf Data Availability} Data sharing is not applicable to this article as it has no associated data.

%\bibliographystyle{plain}
%\bibliography{coulomb}
\def\cprime{$'$}

\end{document}